\newif\ifcomments
\newif\ifstoc
\newif\ifanonymous
\newif\ifacmart
\newif\ifllncs
\newif\ifunitaryoracle

\commentstrue

\ifllncs
    \documentclass[runningheads]{llncs}
\else
    \documentclass{article}
    \usepackage{fullpage}
    \usepackage{amsthm}
    \fi

\ifacmart
  
\else

\usepackage{iftex}
\ifPDFTeX
  \usepackage[utf8]{inputenc}
  \usepackage[noTeX]{mmap}
  \usepackage[T1]{fontenc}
\fi
\ifLuaTeX
  \usepackage{luatex85}
\fi

\usepackage{amssymb}
\usepackage{xcolor}
\definecolor{linkblue}{HTML}{001487}
\usepackage[colorlinks=true,allcolors=linkblue]{hyperref}

\usepackage{authblk}

\usepackage{amsfonts}
\usepackage{mathtools} 
\usepackage{amsmath}
\ifllncs
\else
\usepackage{amsthm}
\fi
\usepackage{graphicx}
\usepackage{comment}
\usepackage{qtree}
\usepackage{tree-dvips}
\usepackage{float}
\usepackage[nameinlink,noabbrev]{cleveref}
\usepackage{braket}
\usepackage{mathrsfs}
\usepackage{tikz}
\usepackage{qcircuit}
\usepackage{xspace}
\usepackage{orcidlink}
\usepackage{dsfont}
\usepackage[inline]{enumitem} 
\ifstoc
  \setlist[description]{noitemsep}
  \setlist[enumerate]{noitemsep}
  \setlist[itemize]{noitemsep}
\fi
\usepackage{longfbox}
\usepackage{csquotes}
\usepackage{mdframed}
\usepackage{microtype}
\usepackage{framed}
\usepackage[colorinlistoftodos]{todonotes}
\let\oldtodo\todo
\renewcommand{\todo}[1]{\oldtodo{\underline{\textbf{ToDo}}:\\#1}}

\ifcomments
  \newcommand{\john}[1]{{{\color{teal}\normalfont [[\bf John: #1]]}}}
\else
  \newcommand{\john}[1]{\ignorespaces}
\fi
\ifcomments
  \newcommand{\barak}[1]{{{\color{olive}\normalfont [[\bf Barak: #1]]}}}
\else
  \newcommand{\barak}[1]{\ignorespaces}
\fi

\ifcomments
  \newcommand{\boyang}[1]{{{\color{blue}\normalfont [[\bf Boyang: #1]]}}}
\else
  \newcommand{\boyang}[1]{\ignorespaces}
\fi

\ifllncs

\renewenvironment{proof}[1][Proof]
{\par\noindent\textit{#1. }}
{\hfill$\square$\par}

\else
\theoremstyle{plain}
\ifacmart
\newtheorem{theorem}{Theorem}
\else
\newtheorem{theorem}{Theorem}[section]
\fi
\newtheorem*{theorem*}{Theorem}
\newtheorem{definition}[theorem]{Definition}
\newtheorem{lemma}[theorem]{Lemma}
\newtheorem{claim}[theorem]{Claim}
\Crefname{claim}{Claim}{Claims}
\newtheorem*{lemma*}{Lemma}

\newtheorem{corollary}[theorem]{Corollary}
\newtheorem*{corollary*}{Corollary}

\theoremstyle{definition}

\theoremstyle{remark}
\newtheorem{remark}[theorem]{Remark}
\fi

\newtheorem{algorithm}{Algorithm}

\crefname{step}{step}{steps}
\Crefname{step}{Step}{Steps}

\usepackage{stmaryrd} 
\makeatletter
\renewcommand{\paragraph}{%
  \@startsection{paragraph}{4}%
  {\z@}{2.25ex \@plus 1ex \@minus .2ex}{-1em}%
  {\normalfont\normalsize\bfseries}%
}
\makeatother
\newcommand{\poly}{\mathrm{poly}}

\newcommand{\negl}{\mathrm{negl}}

\newcommand{\eps}{\epsilon}

\newcommand{\OWF}{\mathsf{OWF}}
\newcommand{\PRS}{\mathsf{PRS}}
\newcommand{\EFI}{\mathsf{EFI}}
\newcommand{\OWSG}{\mathsf{OWSG}}
\newcommand{\onePRS}{\mathsf{1PRS}}

\newcommand{\ketbra}[2]{\ket{#1}\!\!\bra{#2}}

\newcommand{\N}{\mathbb{N}}

\newcommand{\E}{\mathop{\mathbb{E}}}
\newcommand{\Tr}{\mathrm{Tr}}

\newcommand{\reg}[1]{\mathsf{#1}}

\newcommand{\Haar}{\mathrm{Haar}}

\newcommand{\td}{\mathrm{td}}

\newcommand{\fidelity}{\mathrm{F}}

\newcommand{\setft}[1]{\textnormal{#1}}
\newcommand{\id}{\setft{id}}

\newcommand{\bits}{\ensuremath{\{0, 1\}}}
\newcommand{\linear}{\mathrm{L}}

\usepackage{mleftright}

\newcommand{\mbracket}[1]{\mleft[#1\mright]}
\newcommand{\abs}[1]{\left|#1\right|}

\newcommand{\ot}{\ensuremath{\otimes}}

\newcommand{\tr}[1]{\mathrm{Tr}\mbracket{#1}}

\DeclareMathOperator{\pos}{Pos}

\let\1\relax
\newcommand{\1}{\mathds{1}}

\newcommand{\states}{\setft{S}}

\newcommand{\proj}[1]{\ket{#1}\!\!\bra{#1}}

\newcommand{\owpuzz}{\sf{OWPuzz}}

\newcommand{\calA}{\mathcal{A}}
\newcommand{\ncopy}{10000}

\newcommand{\sfA}{\sf{A}}
\newcommand{\sfB}{\sf{B}}

\newcommand{\KeyGen}{\mathsf{KeyGen}}
\newcommand{\StateGen}{\mathsf{StateGen}}

\newcommand{\Ver}{\mathsf{Ver}}

\newcommand{\oracle}{\mathcal{O}}

\usepackage{soul}
\usepackage{xcolor}
\usepackage{xparse}
\makeatletter
\ExplSyntaxOn
\cs_new:Npn \white_text:n #1
{
  \fp_set:Nn \l_tmpa_fp {#1 * .01}
  \llap{\textcolor{white}{\the\SOUL@syllable}\hspace{\fp_to_decimal:N \l_tmpa_fp em}}
  \llap{\textcolor{white}{\the\SOUL@syllable}\hspace{-\fp_to_decimal:N \l_tmpa_fp em}}
}
\NewDocumentCommand{\whiten}{ m }
{
  \int_step_function:nnnN {1}{1}{#1} \white_text:n
}
\ExplSyntaxOff

\NewDocumentCommand{ \varul }{ D<>{5} O{0.2ex} O{0.1ex} +m } {%
  \begingroup
  \setul{#2}{#3}%
  \def\SOUL@uleverysyllable{%
    \setbox0=\hbox{\the\SOUL@syllable}%
    \ifdim\dp0>\z@
    \SOUL@ulunderline{\phantom{\the\SOUL@syllable}}%
    \whiten{#1}%
    \llap{%
      \the\SOUL@syllable
      \SOUL@setkern\SOUL@charkern
    }%
    \else
    \SOUL@ulunderline{%
      \the\SOUL@syllable
      \SOUL@setkern\SOUL@charkern
    }%
    \fi}%
  \ul{#4}%
  \endgroup
}
\makeatother

\usepackage[numbers,sort&compress]{natbib}

\title{Oracle Separation Between Quantum Commitments and Quantum One-wayness}
\date{}
\ifllncs
    \ifanonymous
        \author{}
        \institute{}
    \else
        \author{John Bostanci\inst{1}\orcidlink{0000-0001-9666-7114}  
        \and Boyang Chen\inst{2} \orcidlink{0009-0000-1043-0977}
        \and Barak Nehoran\inst{3}\orcidlink{0000-0001-7371-0829}}
        \institute{
        Columbia University 
        \and Tsinghua University
        \and Princeton University 
        }
    \fi
\else
    \author[1]{John Bostanci}
    \affil[1]{Columbia University}
    \author[2]{Boyang Chen}
    \affil[2]{Tsinghua University}
    \author[3]{Barak Nehoran}
    \affil[3]{Princeton University}
\fi

\begin{document}
\maketitle

\begin{abstract}
    We show that there exists an oracle relative to which quantum commitments exist but no (efficiently verifiable) one-way state generators exist. Both have been widely considered candidates for replacing one-way functions as the minimal assumption for cryptography—the weakest cryptographic assumption implied by all of computational cryptography. Recent work has shown that commitments can be constructed from one-way state generators, but the other direction has remained open. Our results rule out any black-box construction, and thus settles this crucial open problem, suggesting that quantum commitments (as well as its equivalency class of EFI pairs, quantum oblivious transfer, and secure quantum multiparty computation) appear to be strictly weakest among all known cryptographic primitives.
\end{abstract}

\vfill
\noindent \textbf{Note}: There is a bug in \Cref{sec:unitary_oracle_sep}, pointed out to us by Eli Goldin and Mark Zhandry. While  \Cref{lem:swap_from_samples} is true, it only holds for fixed states $\rho$, and thus can not be applied to get an oracle separation between $\EFI$ pairs and $\OWSG$ in the Haar random swap model. The construction of $\owpuzz$ and $\EFI$ can be lifted to the unitary model but the attack on $\OWSG$ cannot be lifted with the approach in this paper. Treating the swap oracle as a reflection around $\ket{0} - \ket{\psi}$ for a Haar random $\ket{\psi}$ allows one recover the result, see \cite{goldin2025translating} for a more detailed discussion.  We leave the buggy attack on $\OWSG$ in that section of the paper in the arXiv version of this paper for future reference.

\newpage

\ifllncs
\else
\tableofcontents
\fi
\newpage
\section{Introduction}

In classical cryptography, one-way functions (OWF) serve as a minimal assumption. That is to say, the existence of nearly any other classical cryptographic primitive implies the existence of one-way functions. Furthermore, many cryptographic primitives (termed ``Minicrypt primitives''), such as pseudorandom generators (PRG), pseudorandom functions (PRF), secret-key encryption and authentication, digital signatures, efficient-far-indistinguishable distributions (EFID), and commitments, are equivalent to one-way functions.

Many of these Minicrypt primitives can be generalized to the setting of quantum states, producing fully quantum primitives such as pseudorandom unitaries (PRU) and pseudorandom states~(PRS)~\cite{Ji_2018}, one-way state generators~(OWSG)~\cite{Morimae_2022,morimae2022one}, one-way puzzles~(OWPuzz)~\cite{cryptoeprint:2023/1620}, efficient-far-indistinguishable quantum state pairs~(EFI)~\cite{brakerski2022computational}, and quantum bit commitments~(QBC)~\cite{eurocrypt-2000-2315}.
A recent sequence of works has shown that although these primitives can be built from one-way functions~\cite{Ji_2018,morimae2022one,cryptoeprint:2023/1620,brakerski2022computational,Morimae_2022}, they may exist even if one-way functions do not~\cite{kretschmer2021,KQST2023quantumcryptographyalgorithmica,lombardi2023onequerylowerboundunitary}.
The classical versions of these fully quantum primitives are known to be equivalent and jointly minimal for classical cryptography, and so a central question in quantum cryptography is whether the same is true for the quantum generalizations.
\begin{center}
    \emph{What is the minimal computational assumption for quantum cryptography?}
\end{center}

In particular, one-way state generators (which generalize one-way functions) and quantum commitments (equivalent to EFI pairs by~\cite{brakerski2022computational}) have received much attention, as the potential minimal assumptions for fully quantum cryptography, with most other such primitives implying one~\cite{Ji_2018,morimae2022one,cryptoeprint:2023/543} or the other~\cite{Morimae_2022,ananth2022cryptography,qian2024hardquantumextrapolationsquantum}. In the classical setting, one-way functions and classical commitments are equivalent and jointly minimal. This motivates the crucial question:
\begin{center}
\emph{Are commitments and one-wayness equivalent in the quantum setting?}
\end{center}

Recent work has given a partial answer by showing that quantum commitments can be constructed from OWSG~\cite{cryptoeprint:2023/1620,batra2024commitmentsequivalentonewaystate}. 
However, the other direction\textemdash showing whether OWSG can be constructed from quantum commitments\textemdash has remained open.%
\footnote{Note that~\cite{batra2024commitmentsequivalentonewaystate} shows that a variant of OWSG that allows verification to be inefficient \emph{is}, in fact, equivalent to quantum commitments. However, it is not known how to construct the standard version of OWSG which requires verification to be efficient.}
We resolve this by showing that OWSG cannot be constructed from quantum commitments in a black-box way. That is, we give a unitary quantum oracle relative to which quantum commitments exist, but every OWSG is insecure.

\begin{theorem}[\textit{informal}]\label{thm:main-informal}
    There is no black-box construction of (efficiently verifiable) one-way state generators from quantum bit commitments.
\end{theorem}

As a direct consequence, we also rule out black-box constructions of a large collection of primitives that imply efficiently verifiable one-way puzzles and one-way state generators.
In fact, our main theorem is stronger. Since our separating oracle allows us to build two stronger cryptographic primitives\textemdash 
(inefficiently verifiable) one-way puzzles,
and 
a single-copy version of pseudorandom states%
\textemdash 
both of which are separately known to imply quantum commitments~\cite{cryptoeprint:2023/1620,Morimae_2022},
this gives us a stronger separation:

\begin{theorem}[\textit{informal}]\label{thm:1prs-vs-owsg}
    There is no black-box construction of (efficiently verifiable) one-way state generators from either single-copy pseudorandom states or (inefficiently verifiable) one-way puzzles.
\end{theorem}

\paragraph{The efficiency of verifying one-wayness.}
In classical cryptography, one-wayness is inherently \emph{verifiable}: that is, given a successful inversion of a one-way function, we can always run the function in the forward direction to check if the inversion was correct.
On the other hand, the literature on \emph{quantum} one-wayness distinguishes between efficiently verifiable~\cite{morimae2022one,goldin2024central} and inefficiently verifiable (or ``statistically verifiable'')~\cite{batra2024commitmentsequivalentonewaystate,cryptoeprint:2023/1620} versions.
This is because there is no built-in way to verify the inversion of a quantum operation that traces out some registers.%
\footnote{One-way state generators in their most general form can output mixed states, while one-way puzzles can sample a puzzle by measuring a quantum register. Both can be implemented by a unitary operation followed by the tracing out of some subset of registers.}
However, it has remained unclear whether these two versions of quantum one-wayness are\textemdash as is the case in the classical setting\textemdash in fact equivalent.%
\footnote{
    \cite{goldin2024central} observe that an oracle separation between efficiently verifiable and statistically verifiable one-way puzzles follows from~\cite{kretschmer2021}. However, the question for the more fundamentally quantum one-way state generators has remained open.
}
\begin{center}
\emph{Can \textbf{efficiently verifiable} quantum one-wayness be constructed from \textbf{statistically verifiable} quantum one-wayness?}
\end{center}
As a corollary to our main theorem, we are able to answer this in the negative:
\begin{corollary}[\textit{informal}]\label{cor:sv-ev-separation-informal}
    There is no black-box construction of either efficiently verifiable one-way state generators or efficiently verifiable one-way puzzles from either statistically verifiable one-way state generators or statistically verifiable one-way puzzles.
\end{corollary}
In other words, in the quantum setting, one-wayness is \emph{not} efficiently verifiable inherently. Efficiently verifiable one-wayness is a stronger assumption.


\paragraph{Conceptual impact of our results.}
Our results show that quantum commitments\textemdash together with their equivalence class of EFI pairs, quantum oblivious transfer, secure quantum multiparty computation, and statistically verifiable one-way state generators\textemdash are strictly weaker than nearly all other known computational cryptography. This motivates defining a new \emph{world} in the spirit of Impagliazzo~\cite{impagliazzo1995personal}. Impagliazzo defines five possible worlds, including \mbox{Cryptomania} (in which classical public-key cryptography exists), and Minicrypt (in which only one-way functions exist). (The three remaining worlds\textemdash Pessiland, Heuristica, and Algorithmica\textemdash do not allow classical cryptography.) The recent work on quantum cryptography has spoken of a Microcrypt, in which one-way function do not exist, but pseudorandom unitaries and states exist (and consequently many other quantum cryptographic primitives). 

We suggest the introduction of a new world to the Impagliazzo hierarchy, \mbox{\emph{Entanglementia}}, a world in which only the bare minimum of (quantum) cryptography is possible, and the only secure computational cryptography that exists is the cryptography that is equivalent to quantum commitments.
We propose the name \emph{Entanglementia}%
\footnote{The ending -mentia means ``in the mind''.}
because it is a world in which the central cryptographic protocols\textemdash such as quantum commitments, oblivious transfer, and secure multiparty computation\textemdash seem to inherently require parties to maintain coherent entanglement between them. Specifically, verification in Entanglementia often requires a challenger to maintain a register that is coherently entangled with the adversary. Entanglementia primitives and assumptions that do not maintain entanglement\textemdash such as statistically verifiable OWSG and EFI\textemdash are inherently not efficiently verifiable.

\subsection{Open Problems}
We suggest the following open problems for future work:

\begin{enumerate}
    \item Are there black-box constructions of pure-output $\OWSG$s from efficiently verifiable mixed-output $\OWSG$s?
    We show that there is no black-box construction of efficiently verifiable $\OWSG$s from statistically verifiable $\OWSG$s. This distinction between efficient/statistical verifiability is only meaningful for $\OWSG$s that produce mixed states, since any $\OWSG$ that produces pure states can be efficiently verified using a SWAP test. This suggests that pure-output $\OWSG$s are qualitatively different. Can a black-box reduction be ruled out?
    \item This work suggests that $\EFI$ pairs, quantum commitments, and their equivalency class of Entanglementia primitives appear to be uniquely minimal among the known computational assumptions for (quantum) cryptography. 
    Can $\EFI$ pairs be constructed from \emph{all} of computational cryptography? Or are there computational assumptions that are even weaker than $\EFI$ but still useful for some cryptography?
    \item Single-copy pseudorandom states with output longer than key ($\onePRS$) are known to imply quantum commitments~\cite{Morimae_2022}. Furthermore, $\onePRS$ appears to be a weak primitive: it can be built from $\OWF$~\cite{Morimae_2022} or $\PRS$~\cite{gunn2023commitments}, but \emph{does not} imply $\OWF$~\cite{KQST2023quantumcryptographyalgorithmica}, $\PRS$~\cite{chen2024power}, or even $\OWSG$~[\textit{this work}] in a black-box way. 
    Similarly, (inefficiently verifiable) one-way \mbox{puzzles} ($\owpuzz$) are known to be implied by pure-output $\OWSG$ and imply quantum commitments~\cite{cryptoeprint:2023/1620}.
    However, the status of $\onePRS$ and $\owpuzz$ is unclear: Can $\onePRS$ and $\owpuzz$ be shown to be separated from commitments and therefore be stronger cryptographic primitives, or are they also contained in Entanglementia?
\end{enumerate}

\subsection{Related Work}

Since the works of \cite{kretschmer2021,KQST2023quantumcryptographyalgorithmica}, which showed that there is an oracle relative to which pseudorandom unitaries exist and one-way functions do not, there has been a large body of work identifying and separating candidates for a minimal assumption for quantum cryptography.  Among these are pseudorandom states~\cite{Ji_2018}, EFI pairs~\cite{brakerski2022computational}, and one-way state generators~\cite{morimae2022one}.  \cite{ananth2022cryptography} showed that pseudorandom states implied EFI pairs, and a number of other useful cryptographic primitives like secure multi-party computation and pseudo one-time encryption schemes.  The work of \cite{khurana2024commitments} further showed that pure-output one-way state generators implied one-way puzzles which in turn implied EFI pairs and commitments. 
\cite{batra2024commitmentsequivalentonewaystate} further showed that mixed-output one-way state generators implied EFI by showing that if the efficient verification of one-way state generators was removed, they would become equivalent to EFI pairs, leaving the question of whether efficiently verifiable one-way state generators could be constructed from EFI pairs or one-way puzzles.


Oracle separations are widely developed for separating classical cryptographic primitives, but fewer techniques are known for separating cryptographic primitives in the fully quantum setting. \cite{kretschmer2021} gave a quantum oracle relative to which pseudorandom unitaries exist but one-way functions do not. Relatedly, \cite{KQST2023quantumcryptographyalgorithmica} shows a classical oracle relative to which single-copy pseudorandom states exist but one-way functions do not.
\cite{goldin2024central} observes that the oracle of \cite{kretschmer2021} can be used to separate one-way puzzles from efficiently verifiable one-way puzzles.
\cite{coladangelo2024black} used a similar oracle to show that pseudorandom states can not be used to build digital signatures for classical messages with quantum public keys, and \cite{goldin2024countcryptquantumcryptographyqcma} uses a modified version of the oracle to separate quantum cryptographic primitives broken by a $\mathsf{QCMA}$ oracle from those broken by a $\mathsf{PP}$ oracle. 

The common reference quantum state model, was introduced by 
\cite{morimae2024unconditionally,qian2024unconditionally}, where they show how to build EFI and quantum commitments.
\cite{chen2024power,ananth2024cryptography} introduced the common Haar random state model, and 
\cite{chen2024power} used this model to show that single-copy pseudorandom states do not imply multi-copy pseudorandom states in a black-box way.

\subsection{Concurrent Work}
Behera, Malavolta, Morimae, Mour, and Yamakawa%
\ifunitaryoracle
\else
~\cite{behera2024oracle} 
\fi
independently and concurrently demonstrate a result similar to ours. Similarly to our result, they show an oracle separation between quantum commitments and both OWSG and efficiently verifiable one-way puzzles.  Our full set of results is, in some sense, incomparable.  We additionally show that 1PRS is separated from OWSG, and they additionally show that primitives such as private-key quantum money are separated from QEFID pairs (\emph{classical} EFI distribution pairs that are quantum-samplable).  We note that they do not consider the common Haar random state model, instead defining a different quantum reference state, and therefore have different proof techniques.  

\ifanonymous
\else
\subsection*{Acknowledgements}
The authors thank Prabhanjan Ananth for helpful discussions about recent results on the common Haar random state model, Fermi Ma for suggesting a new interpretation of the main result of this paper, and Rahul Jain for giving the authors insights into the impact and implications of this work.  The authors also thank Eli Goldin, Henry Yuen, and Mark Zhandry, for helpful discussions related to the Haar random swap oracle, and Amit Behera, Giulio Malavolta, Tomoyuki Morimae, Tamer Mour, and Takashi Yamakawa for their helpful discussions related to their concurrent work.

J.B. is supported by Henry Yuen's AFORS (award FA9550-21-1-036) and NSF CAREER (award CCF2144219).
B.C. acknowledges supported by National Key Research and Development Program of China
(Grant No.\ 2023YFA1009403) and National Natural Science Foundation of China (Grant
No.\ 12347104).
This work was done in part while B.N. and J.B. were visiting the Simons Institute for the Theory of Computing, supported by NSF QLCI Grant No. 2016245.

\fi
\section{Technical Overview}

Our main technical contributions are a polynomial-space attack against one-way state generators relative to all quantum reference quantum state models%
,
as well as a construction of one-way puzzles in the common Haar random state model%
.
Together, they give our separation.

\paragraph{Ruling out one-way state generators relative to common reference states.}
To rule out one-way state generators relative to any common reference quantum state oracle, we notice that the quantum OR attack used in \cite{chen2024power} can be extended to a so-called ``threshold search'' attack.  A threshold search algorithm takes as input a set of $m$ measurements $M$ and $O(\log^2 m \log n)$ copies of a quantum state, and outputs any measurement that has greater than $1/3$ chance of accepting, promised that there exists one that is accepted with probability at least $3/4$.  For one-way state generators, the measurement corresponding to $k$ is to simply run verification with key $k$ on $O(\lambda)$ copies of the input state.  By the correctness of the one-way state generator, the promise of threshold search is satisfied.  Furthermore, because we are taking $O(\lambda)$ copies of the input state, a measurement that accepts with probability $1/3$ means that the input state passes verification with key $k'$ with probability $1 - O(1/\lambda)$.  Thus, this attack breaks the one-wayness of any one-way state generator.  

To implement this attack in polynomial space, we observe that the algorithm from \cite{watts2024quantum}, combined with a space efficient pseudo-random generator from \cite{girish2021eliminating}, provides a $\mathsf{UnitaryPSPACE}$ implementation of threshold search.

We further note that in some common reference state models, such as those of \cite{morimae2024unconditionally,qian2024unconditionally}, or the common Haar random state model of \cite{chen2024power,ananth2024cryptography}, it has been shown that quantum commitments and $\EFI$ pairs exist even relative to adversaries that have unbounded computation\textemdash but a polynomial number of samples of the common reference state.  Thus, for these reference states we arrive at a separation between $\EFI$ pairs and one-way state generators relative to state preparation oracles.\footnote{The result of \cite{chen2024power} additionally constructs single-copy psuedo-random states ($\onePRS$), so our results also imply a separation between them and one-way state generators.}

\paragraph{Constructing one-way puzzles in the common Haar random state model.}
Our second result is to construct inefficiently verifiable one-way puzzles in the CHRS model, strengthening the results of \cite{chen2024power,ananth2024cryptography}, which construct $\onePRS$ and EFI pairs.  Our construction goes as follows:
\begin{enumerate}
    \item $\mathsf{Samp}^{\{\ket{\psi_\ell}\}_\ell}(1^n)$: $\mathsf{Samp}$ first samples a random $n$-bit string $k$. We make use of $n$ independent Haar random states by taking the Haar random states from sizes $n$ through $2n$.%
    \footnote{
        The Haar random state model provides a single Haar random state for every integer $\ell$ (the equivalent of the input size for the oracle). We choose the set of sizes from $n$ through $2n$, although any $\omega(\log n)$ sufficiently large values would work in order to get negligible LOCC Haar indistinguishability.
    }
    Based on the $(\ell-n)$'th bit of the key, $\mathsf{Samp}$ either applies a Pauli $Z$ to the first qubit of $\ket{\psi_{\ell}}$, or does not (i.e. $(Z^{k_{\ell - n}}_{1}) \ket{\psi_{\ell}}$), and then takes the classical shadows of the resulting state.  Let $s$ be the collection of classical shadows generated for $\ell$ from $n$ to $2n$, then $\mathsf{Samp}$ outputs the pair $(k, s)$.
    \item $\mathsf{Ver}^{\{\ket{\psi_\ell}\}_\ell}(k, s)$: $\mathsf{Ver}$ first performs tomography of $\ket{\psi_\ell}$ for $\ell$ from $n$ through $2n$, then uses the classical shadows to estimate the value of the observable $(Z_1^{k_{\ell - n}}) \proj{\psi_\ell} (Z_1^{k_{\ell - n}})$ on the puzzle, which is the overlap with the state that the sampler should have used to generate the classical shadows, if they output key $k$.  If more than $3n/4$ of the observables have value higher than $1/2$, the verifier accepts, otherwise it rejects.
\end{enumerate}

For correctness, if the verifier receives a pair $(k, s)$ that comes from the sampler, than the expected value of the observables is $1$, so with high probability the shadow estimation outputs a value greater than $1/2$.  To prove security, we first consider an adversary that does not sample the Haar random state at all.  In this case, the states $Z^{k_{\ell - n}}_{1} \ket{\psi_\ell}$ are also distributed according to the Haar measure, so the adversary outputs a fixed distribution over solutions to the puzzle, independent of the key $k$.  Therefore, any adversary that does not request samples of the common Haar random state only passes verification with negligible probability.  Finally, due to the LOCC Haar indistinguishability result of \cite{ananth2024cryptography}, an adversary that gets sample access to the common Haar random state outputs (almost) the same distribution over keys as an adversary that gets samples of an indepedently random family of states.  However, this second adversary can be simulated without querying the common Haar random state family at all, and therefore must fail to pass verification.  

\section{Preliminaries}
\subsection{Quantum Basics}
For a bit string $x \in \bits^*$, we denote by $|x|$ its length (not its Hamming weight).
When $x$ describes an instance of a computational problem, we will often use $\lambda = |x|$ to denote its size.

A function $\delta:\N \to [0,1]$ is an \emph{inverse polynomial} if there exists a polynomial $p$ such that $\delta(n) \leq 1/p(n)$ for all sufficiently large $n$. A function $\eps:\N \to [0,1]$ is \emph{negligible} if for every polynomial $p(n)$, for all sufficiently large $n$ we have $\eps(n) \leq 1/p(n)$. 

A \emph{register} $\reg{R}$ is a named finite-dimensional complex Hilbert space. If $\reg{A}, \reg{B}, \reg{C}$ are registers, for example, then the concatenation $\reg{A} \reg{B} \reg{C}$ denotes the tensor product of the associated Hilbert spaces. We abbreviate the tensor product state $\ket{0}^{\ot n}$ as $\ket{0^n}$. For a linear transformation $L$ and register $\reg R$, we write $L_{\reg R}$ to indicate that $L$ acts on $\reg R$, and similarly we write $\rho_{\reg R}$ to indicate that a state $\rho$ is in the register $\reg R$. We write $\Tr[\cdot]$ to denote trace, and $\Tr_{\reg R}[\cdot]$ to denote the partial trace over a register $\reg R$.

We denote the set of linear transformations on $\reg R$ by $\linear(\reg R)$, and linear transformations from $\reg R$ to another register $\reg S$ by $\linear(\reg R, \reg S)$. We denote the set of positive semidefinite operators on a register $\reg{R}$ by $\pos(\reg{R})$. 
The set of density matrices on $\reg R$ is denoted $\states(\reg R)$.
For a pure state $\ket\varphi$, we write $\varphi$ to denote the density matrix $\ketbra{\varphi}{\varphi}$. We denote the identity transformation by $\id$.
For an operator $X \in \linear(R)$, we define $\| X \|_\infty$ to be its operator norm, and $\| X\|_1 = \Tr[|X|]$ to denote its trace norm, where $|X| = \sqrt{X^\dagger X}$.
We write $\td(\rho,\sigma) = \frac{1}{2} \| \rho - \sigma \|_1$ to denote the trace distance between two density matrices $\rho,\sigma$, and $\fidelity(\rho,\sigma) = \| \sqrt{\rho} \sqrt{\sigma} \|_1^2$ for the fidelity between $\rho,\sigma$.

\subsection{The Haar Measure}

Here we state the definition of the Haar measure and Haar random states.

\begin{definition}[Haar measure and Haar random states]
     The Haar measure is the unique left- and right- invariant probability measure on the unitary group $U(d)$.  A Haar random state is sampled by applying a unitary sampled from the Haar measure to $\ket{0}$ (although any initial state would yield the same distribution over states). We use the notation $\Haar(d)$ to refer to the distribution over $d$-dimensional states drawn from the Haar measure.
\end{definition}

One of the most useful properties of the Haar measure is its concentration properties.
\begin{lemma}[Concentration of Haar measure~\cite{meckes2019random}]
\label{lem:haar_concentration}
    Let $N \in \mathbb{N}$.  Let $\mu$ be the Haar measure on dimension $N_i$, and let $\mu$ be the Haar measure on a $N$-dimensional space.  Let $f$ be $L$-Lipshitz function in the Frobenius norm, mapping $N$-dimensional unitaries to real numbers.  Then the following holds for every $t > 0$:
    \begin{equation*}
        \Pr_{U \leftarrow \mu} \left[f(U) \geq \mathop{\mathbb{E}}_{V \leftarrow \mu} \left[f(V)\right] + t\right] \leq \exp\left(-\frac{(\min_{i} \{N_i\} - 2)t^2}{24L^2}\right)\,.
    \end{equation*}
\end{lemma}

We use this for one important corollary, which is the following.
\begin{corollary}[Haar random states on trace $0$ observables]
\label{cor:haar_random_trace_0}
    Let $\ket{\psi}$ be a $n$-qubit Haar random state and $O$ be a trace $0$ observable.  Then the following holds: 
    \begin{equation*}
        \Pr_{\ket{\psi} \leftarrow \Haar(2^{n})}\left[\bra{\psi} O \ket{\psi} \geq 2^{-\sqrt{n/2}}\right] \leq \exp\left(-\frac{2^{n/2}-2}{96}\right)\,.
    \end{equation*}
\end{corollary}
\begin{proof}
    We note that we can equivalently phrase this probability as being over a Haar random unitary, and the state $U \ket{0}$.  The function $f(U) = \bra{0}U^{\dagger} O U \ket{0}$ is a degree $2$-polynomial in $U$, so $f$ is $2$-Lipshitz in the Frobenius norm.  Applying \Cref{lem:haar_concentration} to this function and $t = 2^{-\sqrt{n/2}}$ yields the desired result. 
\end{proof}

\subsection{Quantum Oracles and the Common Reference States}

\begin{definition}[Quantum oracle access]
    Let $f: \{0, 1\}^{n} \mapsto \{0, 1\}$ be a classical Boolean function, a quantum query algorithm $\mathcal{A}^{(\cdot)}$ queries $f$ via access to a unitary $U_{f}$ that acts as
    \begin{equation*}
        U_{f} \ket{x} \ket{b} \mapsto \ket{x} \ket{b \oplus f(x)}\,.
    \end{equation*}
\end{definition}

Typically quantum oracles must be unitary transformations, however recently a new model of ``isometry'' oracles has appeared in the literature~\cite{chen2024power,ananth2024cryptography}.  This model, which essentially allows access to a very specific quantum resource (a single quantum state and no way to un-compute it) has been shown to allow for oracle separations between $\EFI$ pairs and $\PRS$.  However, the question of boosting these separations to standard oracles remains an open question.

\begin{definition}[Common reference quantum state (CRQS) \cite{morimae2024unconditionally,qian2024unconditionally}]
    The common reference quantum state (CRQS) model is an isometry that can be accessed as a quantum oracle.  Let $\mathcal{V} = \{V_m\}_{m \in \mathbb{N}}$, with $V_m: \mathbb{C} \mapsto \mathbb{C}^{2^m}$ be a family of isometries such that 
    \begin{equation*}
        V_{m} \ket{\alpha} \mapsto \ket{\alpha} \ket{\phi_{m}}\,,
    \end{equation*}
    We call $\{\ket{\phi_{m}}\}$ the family of common reference quantum states.%
    \footnote{\cite{morimae2024unconditionally,qian2024unconditionally} require that the common reference quantum states are efficiently preparable by a third-party setup algorithm. Since the states in our setting are prepared by an oracle, we do not make this requirement. Because of this, while the quantum auxiliary input model and CRQS models of~\cite{morimae2024unconditionally,qian2024unconditionally} are incomparable, both of their models are special cases of the CRQS model by our definition.}
\end{definition}

A special case of the common reference state model is when the reference state family is drawn uniformly at random from the Haar measure.

\begin{definition}[Common Haar random state model (CHRS) \cite{chen2024power,ananth2024cryptography}]
    The common Haar random state (CHRS) model is a CRQS isometry where every $\ket{\psi_{m}}$ is drawn from $\Haar(2^{m})$.
\end{definition}

Secure $\EFI$ pairs (and thus quantum commitments) as well as $\onePRS$ are known to exist relative to a CHRS oracle~\cite{chen2024power,ananth2024cryptography}.

\subsection{Cryptographic Primitives}

Here we define a number of cryptographic primitives, as well as their instantiation in the common Haar random state model.

\begin{definition}[Efficient, far, indistinguishable pairs \cite{brakerski2022computational}]
   A family of pairs of quantum states $\{(\rho_{0, \lambda}, \rho_{1, \lambda})\}_{\lambda}$ is an $\EFI$ pair if
   \begin{itemize}
       \item (Efficiently preparable) There exists a family of polynomial-size, time efficient quantum circuits $\{C_{\lambda}\}_{\lambda}$ such that $\Tr_{\reg{B}}[C_{\lambda}\ket{b0}] = \rho_{b, \lambda}$.
       \item (Statistically far) There exists a negligible function $\mu$ such that 
       \begin{equation*}
       \td(\rho_{0, \lambda}, \rho_{1, \lambda}) \geq 1 - \mu(\lambda)\,.
       \end{equation*}
       \item (Computationally indistinguishable) There exists a negligible function $\nu$ such that for all polynomial-time quantum adversaries $A$, 
       \begin{equation*}
           \left|\Pr[A(\rho_{0, \lambda}) = \top] - \Pr[A(\rho_{1, \lambda}) = \top]\right| \leq \nu(\lambda)\,.
       \end{equation*}
   \end{itemize}
\end{definition}

\begin{definition}[$\EFI$ pairs in the CRQS model]
    For an $\EFI$ pair in a CRQS model, the circuits for preparing the $\EFI$ pair has access to the common reference quantum state oracle, and the computationally indistinguishability holds for all adversaries who have query access to the common reference quantum state oracle.
\end{definition}

\begin{remark}
    We do not formally define quantum bit commitments here, as they are equivalent to $\EFI$ pairs by~\cite{brakerski2022computational}, and this equivalence extends to the unitary oracle setting, which is our \emph{ultimate} goal.
    Interestingly, whether this equivalence holds in under all CRQS models is not clear, since the formal equivalence between different flavors of commitments requires the uncomputing of quantum circuits, which cannot necessarily be done in CRQS models. More specifically it is not known if computationally binding and statistically hiding commitments are equivalent to statistically binding and computationally hiding commitments in general in CRQS models (where by ``computational'' security we mean secure against adversaries that receive only a polynomial number of copies of the CRQS).
    This is of course not an issue for the headline results of our paper, since we ultimately upgrade all our results to a unitary oracle model, in which the equivalences hold once more. Furthermore, $\EFI$ pairs in CRQS models are equivalent to computationally hiding commitments by the techniques of~\cite{morimae2024unconditionally,qian2024unconditionally} (which extend those of ~\cite{brakerski2022computational} to CRQS models).
    For simplicity, we will therefore mainly only refer to $\EFI$ pairs for the rest of this~paper.
\end{remark}

Next we define a one-way state generator, using the definition from \cite{morimae2022one}.

\begin{definition}[One-way state generators \cite{Morimae_2022,morimae2022one}]
    A one-way state generator ($\OWSG$) is a collection of QPT algorithms $(\KeyGen, \StateGen, \Ver)$ such that
    \begin{itemize}
        \item $\KeyGen$ takes as input the security parameter $1^{\lambda}$ and outputs a classical key $k \in \{0, 1\}^{\kappa}$.
        \item $\StateGen$ takes as input a classical key $k$ and outputs a $m$-qubit quantum state $\rho_k$.
        \item $\Ver$ takes as input a classical key $k$ and quantum state $\rho$ and outputs either $\bot$ or $\top$.
    \end{itemize}
    A $\OWSG$ satisfies correctness if for all $\lambda$,
    \begin{equation*}
        \Pr \left [ \mathsf{Ver}(k, \rho_k) \text{ accepts} : \begin{array}{c}
k \leftarrow  \mathsf{KeyGen}(1^\lambda) \\
\rho_k \leftarrow \StateGen(k) 
\end{array}
   \right ] \geq 1 - \negl(\lambda)\,.
    \end{equation*}
    A $\OWSG$ satisfies one-way security if for all polynomial-time quantum adversaries $A$ and polynomials~$t$, 
    \begin{equation*}
        \Pr \left [ \mathsf{Ver}(k, \rho_{k'}) \text{ accepts} : \begin{array}{c}
k' \leftarrow  \mathsf{KeyGen}(1^\lambda) \\
\rho_{k'} \leftarrow \StateGen(k')\\
k \leftarrow A\left(1^{\lambda}, \rho_{k'}^{\otimes t}\right)
\end{array}
   \right ] \leq \negl(\lambda)\,.
    \end{equation*}
\end{definition}

Note that we follow the convention set by the existing literature and define $\OWSG$s to have efficient verification, but not necessarily with pure-state outputs.  

\begin{definition}[One-way state generators in the CRQS model]
    For a one-way state generator in the CRQS model, both $\KeyGen$ and $\StateGen$ have access to the common reference quantum state oracle, and one-way security holds relative to all polynomial-time quantum adversaries that have access to the common reference quantum state oracle.
\end{definition}

\begin{definition}[One-way puzzle \cite{cryptoeprint:2023/1620}]
\label{def:one_way_puzzle}
    A one-way puzzle ($\owpuzz$) is a pair of quantum algorithms $(\mathsf{Samp}, \Ver)$ such that
    \begin{enumerate}
        \item $\mathsf{Samp}$ takes as input a security parameter $1^{\lambda}$ and outputs a pair of classical strings $(k, s)$, where $k \in \{0, 1\}^{\lambda}$.  $\mathsf{Samp}$ must be efficient.
        \item $\Ver$ takes as input a pair $(k, s)$ and outputs either $\bot$ or $\top$.
    \end{enumerate}
    A $\owpuzz$ satisfies correctness if for all $\lambda$,
    \begin{equation*}
        \Pr\left[\Ver(k, s) \text{ accepts}: (k, s) \leftarrow \mathsf{Samp}(1^{\lambda})\right] \geq 1 - \negl(\lambda)\,.
    \end{equation*}
    A $\owpuzz$ satisfies security of for all polynomial-time quantum adversaries $A$
    \begin{equation*}
        \Pr\left[\Ver(A(s), s) \text{ accepts}: (k, s) \leftarrow \mathsf{Samp}(1^{\lambda})\right] \leq \negl(\lambda)\,.
    \end{equation*}
\end{definition}
We require that $\mathsf{Samp}$ is efficient (QPT), but $\Ver$ may be efficient (efficiently verifiable one-way puzzle) or inefficient (statistically verifiable one-way puzzle). We use the convention from prior work of using $\owpuzz$ to mean the inefficiently verifiable version, but we specify which version we mean when we believe it may not be clear from context.

\begin{definition}[One-way puzzles with sample-efficient verifier in the CRQS model]
    For a one-way puzzle with sample-efficient verifier in the CRQS model, both $\mathsf{Samp}$ and $\mathsf{Ver}$ can make \emph{polynomial} many calls to the common reference quantum state oracle, and security holds relative to all polynomial-time quantum adversaries that have access to the common reference quantum state oracle.
\end{definition}

\begin{definition}[One-way puzzles with sample-inefficient verifier in the CRQS model]
    A one-way puzzle with sample-inefficient verifier is a pair of sampling and verification algorithms $(\sf{Samp}, \sf{Ver})$ with the same syntax as \Cref{def:one_way_puzzle}, except that 
    $\Ver(k, s, \bigotimes_{i=1}^r \ket{\psi_i}^{\otimes r}) \to \top/\bot$ is a time-unbounded algorithm that on input of any pair classical strings $(k,s)$ halts and outputs $\top/\bot$, where $r=r(n)$ can be arbitrarily function of $n$.
\end{definition}

Note that since $\Ver$ is allowed to be unbounded in a one-way puzzle, there is a distinction between one-way puzzles and sample-efficient one-way puzzles in the CRQS model.  In order to truly rule out one-way puzzles in a CRQS model, one should rule out one-way puzzles in the CRQS model that have unbounded query access to the reference state.

\begin{definition}[Single-copy pseudo-random states \cite{Morimae_2022}]
    A single-copy pseudo-random states generator $\onePRS$ is a QPT algorithm $\mathsf{Gen}$ that takes as input a key $k \in \{0, 1\}^{\lambda}$ of length $\lambda$ and outputs a pure state $\ket{\psi_{k}}$ on $m(\lambda) > \lambda$ qubits.  
    
    A $\onePRS$ satisfies the pseudo-randomness property if for all polynomial-time quantum adversaries $A$ and $\lambda \in \mathbb{N}$,
    \begin{equation*}
        \left|\Pr_{k \leftarrow \{0, 1\}^{\lambda}}\left[A(\ket{\psi_{k}})\text{ accepts}\right] - \Pr_{\ket{\psi} \leftarrow \Haar(2^{m(\lambda)})}\left[A(\ket{\psi}) \text{ accepts}\right]\right| \leq \negl(\lambda)\,.
    \end{equation*}
\end{definition}
\begin{definition}[Single copy pseudo-random states in the CRQS model]
    For a single-copy pseudo-random state in the CRQS model, $\mathsf{Gen}$ has access to the common reference quantum state oracle, and pseudo-randomness holds relative to all polynomial-time quantum adversaries that have access to the common reference quantum state.
\end{definition}

\subsection{Unitary Complexity Theory}

\begin{definition}[$\mathsf{unitaryPSPACE}$~\cite{bostanci2023unitary}]
    A unitary synthesis problem $(U_x)_{x}$ is in $\mathsf{UnitaryPSPACE}$ if for all polynomials $p(n)$, there exists a uniform polynomial-space algorithm $C$ that implements $(U_x)_{x}$ with worst-case error $1/p(n)$.  
\end{definition}

An important fact about $\mathsf{unitaryPSPACE}$ is that it has a complete problem (under Turing reductions), the succinct Uhlmann transformation problem~\cite{bostanci2023unitary}.  By oracle access to $\mathsf{unitaryPSPACE}$, we therefore mean oracle access to the succinct Uhlmann transformation unitary synthesis problem, or any other $\mathsf{unitaryPSPACE}$-complete problem.

\subsection{Quantum Learning Theory}

In this section we review results and definitions from quantum learning theory that will be relavent to our result.

\begin{definition}[Threshold search~\cite{badescu2024improved}]
     Let $\{M_i\}_{i \in [m]}$ be a collection of $2$-outcome measurements.  Let $\rho$ be an unknown quantum state with the promise that there exists an index $i$ such that 
     \begin{equation*}
         \Tr[M_i \rho] \geq 3/4\,.
     \end{equation*}
     The threshold search problem is to output a measurement $M_j$ such that $\Tr[M_j \rho] \geq 1/3$.
\end{definition}

\begin{theorem}[Random threshold search \cite{watts2024quantum}]\label{thm:random_threshold_search}
    There is an algorithm that uses $O(\log^2(m))$ space and samples of $\rho$, has expected time $O(m)$, and makes intermediate measurements, that solves the threshold search problem with constant probability.
\end{theorem}

The algorithm from \cite{watts2024quantum} is to fix a threshold $\theta \in [0.4, 0.6]$, and to repeatedly measure a threshold measurement for a randomly sampled $M_i$ on $\log^2(m)$ copies of the state, testing whether more than a $\theta$ percent of the measurements accepted.  

The following theorem allows us to simulate this algorithm in $\mathsf{unitaryPSPACE}$.  
\begin{theorem}[$\mathsf{unitaryPSPACE}$-simulation~\cite{girish2021eliminating}]
    Every quantum algorithm that runs in time $T$ with space $S \geq \log(T)$ with unitary operations and intermediate unital%
    \footnote{Unital measurements are those that output a collapsed quantum state but no classical outcome. Thus they send the fully mixed state to the fully mixed state. In particular, they do not allow discarding a quantum state nor resetting it to a fixed state. This limited kind of measurement, however, is sufficient for our setting in which we only need to sample randomness.}
    measurements can be simulated by a quantum algorithm of time $T \cdot S^2 \poly \log(S)$ and space $O(S \cdot \log T)$ with only unitary operations and no intermediate measurements.
\end{theorem}

This allows us to simulate the threshold search algorithm in $\mathsf{UnitaryPSPACE}$.

\begin{remark}
\label{rem:threshold_search_pspace}
The result of \cite{girish2021eliminating} is proven by providing an unconditional pseudo-random generator with small seed that can be implemented in $\mathsf{PSPACE}$.  Since the algorithm of \cite{watts2024quantum} just samples $O(m)$ random measurements and applies controlled versions of them on the input state (controlled on a single output register that indicates whether any of the previous measurements accepted), it could be applied directly to threshold search to get a simpler simulation of this algorithm.

This also means that a polynomial-time query algorithm can easily generate a succinct circuit representing the algorithm, and we really need \emph{only} an oracle to a $\mathsf{unitaryPSPACE}$-complete problem, \emph{not} the ability to perform any $\mathsf{PSPACE}$ computation ourselves.  This is in contrast to previous work on this subject, which provided an oracle called a ``QPSPACE'' oracle, which essentially made all parties polynomial space computations with intermediate measurements.
\end{remark}
\section{Separation Between OWSG and EFI Pairs}
\label{sec:state_separation}

In this section, we show that there exists an oracle in any common reference state model such that no one-way state generator exists.  Note that, as stated before, \cite{chen2024power,ananth2024cryptography} already showed that in the common Haar random state model, quantum commitments and $\EFI$ pairs exist (additionally, they show that $\onePRS$ exist), and that the security cannot be broken by an adversary of any complexity that only has access to polynomial copies of the common Haar state.
Together, this will imply an oracle in the common reference state model relative to which $\EFI$ pairs exist but one-way state generators do not.

\begin{theorem}\label{thm:one_way_state_generators_do_not_exist}
    Let $\mathcal{V}$ be a common reference quantum state oracle for any family of reference states. Relative to $(\mathcal{V}, \mathsf{unitaryPSPACE})$, efficiently verifiable one-way state generators do not exist. 
\end{theorem}
\begin{proof}
    Let $(\KeyGen, \StateGen, \Ver)$ be a $\OWSG$.  In the common Haar random state model, we can assume that $\Ver$ works by calling some quantum circuit (one for each key $k$) on input $\ket{\psi}_{\reg{A}} \otimes (\ket{\phi_{1}}^{\otimes s} \ldots \ket{\phi_{s}}^{\otimes s})_{\reg{B}}$ for some $s = \poly(\lambda)$, and then measuring a bit in the computational basis.  Let $U_{k}$ be said circuit and consider the following operator that acts on the input for $U_{k}$ (registers $\reg{AB}$), and a copy of a $\OWSG$ state in register $\reg{C}$. We define the measurement $\Pi_{k}$ as follows
    \begin{equation*}
        \Pi_{k} = \left(\big((U_{k}^{\dagger})_{\reg{CD}}\big) \left(\proj{1}_{\reg{C}} \otimes \id_{\reg{B}}\right)\big((U_{k})_{\reg{AB}}\big)\right)^{\otimes 10 \lambda}\,.
    \end{equation*}
    \begin{claim}
    \label{claim:owsg_completeness}
        On quantum input $(\ket{\psi_{k}} \otimes \ket{\phi_{1}}^{\otimes s} \ldots \ket{\phi_{s}}^{\otimes s})^{\otimes 10 \lambda}$, $\Pi_{k}$ accepts with probability $1 - \negl(\lambda)$.
    \end{claim}
    \begin{proof}
        By the correctness of the $\OWSG$, running $\Ver$ on a copy of $\ket{\psi_{k}}$ and key $k$ accepts with probability $1 - \negl(\lambda)$.  Since the probability that $\Pi_k$ accepts is the probability that $10\lambda$ many verifiers (run in parallel) accept, its accept probability is given by
        \begin{equation}
            (1 - \negl(\lambda))^{\otimes 10\lambda} \geq 1 - 10 \lambda \cdot \negl(\lambda) = 1 - \negl(\lambda)\,.
        \end{equation}
        Here the second line is Bernoulli's inequality. This completes the proof of the completeness of the algorithm.
    \end{proof}

    \begin{claim}
    \label{claim:owsg_soundness}
        If $\Pi_{k}$ accepts with probability $\geq 1/3$ for some key $k$, then the probability that verification accepts is at least $1 - \frac{1}{5\lambda}$.
    \end{claim}
    \begin{proof}
        Let $p$ be the probability that $\Ver$ accepts when given key $k$ and state $\ket{\psi}$.  Then it is clear to see that the probability $\Pi_{k}$ accepts on the state $\ket{\psi}^{\otimes 10 \lambda}$ is
        \begin{equation*}
            p^{10 \lambda} \geq \frac{1}{3}\,.
        \end{equation*}
        Solving for $p$, we see that
        \begin{align*}
            p &\geq e^{-\frac{\ln(3)}{10\lambda}}\\
            &\geq 1 - \frac{\ln(3)}{10 \lambda}\\
            &\geq 1 - \frac{1}{5\lambda}\,.
        \end{align*}
        Here we use the inequality $1 - x \leq e^{-x}$, and then we use the fact that $\ln(3) \leq 2$.  This completes the proof that the algorithm always provides a key that violates one-way state generator security.
    \end{proof}

    The algorithm for breaking a one-way state generator is to run threshold search on $O(\lambda^2)$ many copies of the input state, and return the key corresponding to the measurement that threshold search outputs. Note that we need $O(\lambda)$ for every $\Pi_k$ and threshold search requires $O(\log(m)) = O(\lambda)$ copies of the input state, which is itself $O(\lambda)$ copies of the one-way state generator state, to run. 
    
    From the first claim, the promise of threshold search is met, so threshold search outputs a key such that $\Tr[\Pi_{k} \rho] \geq 1/3$ with constant probability.  From the second claim, we know that the key will be accepted by the verifier with probability at least $1 - \frac{1}{5\lambda}$, which contradicts one-way state generator security.  Finally, as noted in \Cref{rem:threshold_search_pspace}, the threshold search algorithm can be implemented with oracle access to a $\mathsf{unitaryPSPACE}$-complete problem. This completes the proof of \Cref{thm:one_way_state_generators_do_not_exist}.
\end{proof}

We can apply our attack to the common Haar random state model to get an oracle separation between $\onePRS$ and one-way state generators.

\begin{theorem}
    In the common Haar random state model, $\onePRS$ exists and one-way state generators do not.
\end{theorem}
\begin{proof}
    \cite{chen2024power,ananth2024cryptography} prove that in the common Haar reference state model, $\onePRS$ exists, and consequently $\EFI$ pairs exist as well, relative to all polynomial-sample adversaries (with unbounded computation otherwise).  Since the common Haar reference state model is a CRQS model, \Cref{thm:one_way_state_generators_do_not_exist} implies that one-way state generators do not exist.
\end{proof}

We also note that since our attack against $\OWSG$s works in \emph{all} common reference quantum state models, a weaker separation, between $\EFI$ pairs and one-way state generators, can be similarly attained with a deterministic oracle if one instead takes the quantum auxiliary input model from~\cite{morimae2024unconditionally,qian2024unconditionally}.%
\footnote{The quantum auxiliary input model, as defined in~\cite{morimae2024unconditionally,qian2024unconditionally} is a special type of CRQS model (by our definition), in which the common reference quantum state for each value of the security parameter (or input size) is a predetermined \emph{fixed} state.}

\subsection{One-Way Puzzles in the Common Haar Random State Model}

In this section, we provide a construction of inefficiently verifiable one-way puzzles in the common Haar random state model.  We also note that our adversary (from \Cref{thm:one_way_state_generators_do_not_exist}) breaks all verification-efficient one-way puzzles in the common Haar random state model.  
\begin{corollary}\label{cor:ev-owpuzz-non-exist}
    Let $\mathcal{V}$ be a common reference quantum state oracle for any family of reference states.  Relative to $(\mathcal{V}, \mathsf{UnitaryALL})$, sample-efficient one-way puzzles do not exist.
\end{corollary}
\begin{proof}
    The adversary for one-way puzzles is similar to the adversary from  \Cref{thm:one_way_state_generators_do_not_exist}, except that instead of requiring $10\lambda$ copies of $\ket{\psi_{k}}$, it simply copies $s$ into $10\lambda$ registers and runs threshold search on $\{\Pi_{k}\}$.  The same proof shows that the adversary will retrieve a key that is accepted by the verifier with non-negligible probability.  
    
    Formally, to get an adversary that calls an oracle, we can define the ``inefficient one-way puzzle verification'' problem, where the instance is a description of the one-way puzzle, the input is a classical pair of strings $\ket{k, s}$, and copies of the common reference quantum state, and the output is the result of the verification.  Since verification exists, this is a problem in $\mathsf{unitaryALL}$.  We further note that our adversary uses the same amount of space as the verifier for the one-way puzzle does, but since verification for a one-way puzzle is not required to be space efficient, our verifier might not be. \emph{If} the verifier happens to be polynomial space, this adversary will also be polynomial space.
\end{proof}

As noted in the discussion, sample-efficient one-way puzzles were already ruled out by the LOCC indistinguishably results of \cite{ananth2024cryptography}, but we believe our proof is simpler and thus might be of independent interest to the reader.  

On the contrary, inefficiently-verifiable $\owpuzz$ exist in the CHRS model, as formally states in the following theorem.

\begin{theorem}\label{thm:iv-owpuzz}
    Let $\mathcal{V}_{\Haar}$ be the common reference quantum state oracle for Haar random states. Relative to $(\mathcal{V}_{Haar}, \mathsf{UnitaryALL})$, inefficiently verifiable one-way puzzles exist.
\end{theorem}

We now present our construction of~\Cref{thm:iv-owpuzz}.  The construction of $\owpuzz$ relies on the classical shadow tomography~\cite{HKP20}, so here we first describe how to sample classical shadows of a state and the theoretical guarantee of the classical shadows algorithm. Assume we have $N$ copies of a state $\rho$ and we want to estimate the observables $O_1, \dots, O_M$ with the copies of $\rho$. One can do this by just doing independent random Clifford measurements over all the copies of $\rho$. Namely, we can sample random Clifford unitaries $C_1, \dots, C_N$, and then measure $C_i\rho C_i^\dagger$ and record their measurement results $b_1, \dots, b_N$. We call the collection of $(b_i, C_i)$ the classical shadows of the state $\rho$, denoted as $\mathsf{ShadowGen}(\rho, N)$, which can be sampled without prior knowledge of the observables. The following lemma show that the classical shadows can be used to estimate the probabilities $\tr[O_i \rho]$.
\begin{lemma}[Classical shadow tomography, adapted from~\cite{HKP20}]\label{lem:classical-shadow-tomography}
    For any $n$-qubit observables $O_1, \dots, O_M$, and accuracy parameter $\eps, \delta \in [0,1]$. Let $N \geq \frac{204}{\eps^2}\log(2M/\delta)\max_{1\leq i \leq M}\tr {O_i^2}$. Then for any $n$-qubit state $\rho$, let $\sf{ShadowGen}(\rho, N)$ be the classical shadow tomography of $N$ copies of $\rho$. Then there is a time-unbounded classical algorithm that can give an estimation $\hat{o}_i$ on all the observables $\tr{O_i \rho}$ given $\sf{ShadowGen}(\rho, N)$ such that
    \begin{equation*}
        |\hat{o}_i - \tr{O_i \rho}| \leq \eps \quad \forall i \in [1,M]\,,
    \end{equation*}
    with probability at least $1-\delta$.
\end{lemma}


\noindent Now we can describe the construction of one-way puzzles in the CHRS model:
\begin{itemize}
    \item The sampler $\sf{Samp}$ takes $\ncopy$ copies of each $l$-qubit CHRS state for $n \leq l \leq 2n$. For each $l \in [n,2n]$, the sampler chooses a random bit $k_l \in \{0,1\}$ and then generates the classical shadow of $(Z^{k_l}_1 \otimes \id )\ket{\psi_l}$ using the $\ncopy$ copies of the state. The sampler will output $(k_n \dots k_{2n})$ as the key and $\mathsf{ShadowGen}((Z^{k_l}_1 \otimes \id)\ket{\psi_l}\!\!\bra{\psi_l}(Z_1^{k_l}\otimes \id))$ as the puzzle.
    \item The verifier $\sf{Ver}$ runs the algorithm from~\Cref{lem:classical-shadow-tomography} for estimating observables for all $l \in [n, 2n]$, for the observable $(Z^{k_l}_1 \otimes \id)\ket{\psi_l}\!\!\bra{\psi_l}(Z^{k_l}_1 \otimes \id)$ and state $\rho = (Z^{k_l}_1 \otimes \id)\ket{\psi_l}\!\!\bra{\psi_l}(Z^{k_l}_1 \otimes \id)$. The verifier accepts if for at least $3n/4$ different $l$, the estimation of the corresponding observable is greater than 1/2. The estimation of $(Z^{k_l}_1 \otimes \id)\ket{\psi_l}\!\!\bra{\psi_l}(Z^{k_l}_1 \otimes \id)$ needs the classical description of $\ket{\psi_l}$, which can be done if the verifier has access to exponentially many copies of $\ket{\psi_l}$.
\end{itemize}

\begin{theorem}[One-way puzzle correctness]
    On input $(k, s) \gets \mathsf{Samp}^{\{\ket{\psi_l}\}}_{l}$, $\mathsf{Ver}$ accepts with probability $1 - \negl(n)$.
\end{theorem}

\begin{proof}
According to~\Cref{lem:classical-shadow-tomography} with $\eps = 1/3$ and $\delta = 1/10$, the probability that the estimate of $Z_1^{k_l} \otimes \id$ deviates by at least $1/2$ is less than $1/10$ when the number of copies of the classical shadow is $N \geq 204 \cdot 9 \cdot \log (2 \cdot 100)$, for which $\ncopy$ copy is enough. Since $(Z^{k_l}_1 \otimes \id)\ket{\psi_l}$ measures 1 on $(Z^{k_l}_1 \otimes \id)\ket{\psi_l}\!\!\bra{\psi_l}(Z^{k_l}_1 \otimes \id)$, the probability that the estimation is lower than 1/2 is at most 1/10. Thus by the Chernoff bound over all $n$ keys, the probability that a puzzle is not accepted by the verifier (i.e. more than $1/4$ of the puzzles have an estimate lower than $1/2$) is exponentially small, which shows $\owpuzz$ correctness.
\end{proof}

At a high level, the security proof will begin by first showing that an adversary $\mathcal{A}$ that receives no copies of the common Haar random state succeeds with probability negligible in $n$.  Then, using the LOCC Haar indistinguishability lemma from \cite{ananth2024cryptography}, we will show that an adversary that receives $\poly(n)$ copies of the Haar random state can not distinguish between the case when they have the same common Haar random state as the sampler, and when they have independently sampled Haar random states.  As the second case can be simulated without sampling any copies of the Haar random state by an exponential time adversary that samples Haar random states and runs $\mathcal{A}$ on them, the probability that the adversary passes $\mathsf{Ver}$ must be negligible in $n$.  Formally, we have the following.

\begin{lemma}
\label{lem:no_samples_cant_verify}
    Let $\mathcal{A}$ be an adversary that gets no samples of the common Haar random state $\{\ket{\psi_l}\}_{l}$, then $\mathcal{A}(s)$, $s \gets \mathsf{Samp}^{\{\ket{\psi_l}\}_{l}}$, passes verification with probability that is negligible in $n$.
\end{lemma}
\begin{proof}
    In order to prove the theorem, we note that by the definition of the Haar measure as the left and right invariant measure on states, for all choices of $k_n, \ldots k_{2n}$,
    \begin{multline*}
        \mathop{\mathbb{E}}_{\{\ket{\psi_l}\}_{l}} \left[\proj{\psi_n}^{\otimes 10000} \otimes \ldots \otimes \proj{\psi_{2n}}^{\otimes 10000}\right] \\= \mathop{\mathbb{E}}_{\{\ket{\psi_l}\}_{l}} \bigg[(Z^{k_n}_1)^{\otimes 10000}\proj{\psi_n}^{\otimes 10000} (Z^{k_n}_1)^{\otimes 10000} \otimes \\ \ldots \otimes (Z^{k_{2n}}_1)^{\otimes 10000}\proj{\psi_{2n}}^{\otimes 10000}(Z^{k_{2n}}_1 )^{\otimes 10000}\bigg]\,.
    \end{multline*}
    Therefore, the two distributions are identical.
    \begin{equation*}
        \mathcal{A}(\mathsf{Samp}^{\{\ket{\psi}_{l}\}_{l}})= \mathop{\mathbb{E}}_{\{k_l\} \gets \{0, 1\}^*} \left[\mathcal{A}(\mathsf{Samp}^{\{X^{k_l}_1\ket{\psi}_{l}\}_{l}})\right]\,.
    \end{equation*}
    Then the probability that the adversary $\mathcal{A}$ passes $l$'th test of the verification is upper bounded by the probability that $\mathcal{A}$ outputs $k_{l}$, plus the probability that the verifier accepts, given that the $l$'th key output by the adversary is not $k_l$.  The first probability is at most $1/2$, and from \Cref{lem:classical-shadow-tomography} with $\delta = 1/10$, the second is at most $1/10$.  Since each $k_l$ is chosen uniformly at random and the distribution over the keys output by $\mathcal{A}$ does not depend on the choice of $k_l$, the probability that the adversary succeeds is at most $(3/5)^{n} = \negl(n)$.
\end{proof}

The final step is to apply the LOCC Haar indistinguishability result of \cite{ananth2024cryptography}, stated formally below. 
\begin{lemma}[LOCC Haar indistinguishability, adapted from~\cite{ananth2024cryptography}]\label{lem:locc-harr-indi}
    For positive integers $s, t, n_1, \dots, n_s$, define
    \begin{align*}
        \rho_{\sfA\sfB} &\coloneqq \bigotimes_{i=1}^s \E_{\ket{\psi_i} \gets \Haar(2^{n_i})} \left[ (\ket{\psi_i}\!\!\bra{\psi_i}^{\otimes t})_{\sfA_i} \otimes (\ket{\psi_i}\!\!\bra{\psi_i}^{\otimes t})_{\sfB_i} \right] \\
        \sigma_{\sfA\sfB} &\coloneqq \bigotimes_{i=1}^s \E_{\ket{\psi_i} \gets \Haar(2^{n_i})} \left[ (\ket{\psi_i}\!\!\bra{\psi_i}^{\otimes t})_{\sfA_i} \right] \otimes \bigotimes_{i=1}^s \E_{\ket{\phi_i} \gets \Haar(2^{n_i})} \left[ (\ket{\phi_i}\!\!\bra{\phi_i}^{\otimes t})_{\sfB_i} \right],
    \end{align*}
    where $\sfA = (\sfA_1, \dots, \sfA_s)$, $\sfB = (\sfB_1, \dots, \sfB_s)$. Then $\rho_{\sfA\sfB}$ and $\sigma_{\sfA\sfB}$ are $O(\sum_{i=1}^s t^2/2^{n_i})$-LOCC indistinguishable.
\end{lemma}
At a high level, \Cref{lem:locc-harr-indi} says that any adversary for the one-way puzzle can not distinguish between the case when they have been given copies of $\ket{\psi_l}$, the states that $\mathsf{Samp}$ used, or $\ket{\psi'_l}$, independently sampled Haar random states.  However, an adversary $\mathcal{A}$ that receives copies of independently sampled $\ket{\psi'_l}$ can be simulated by an (exponential time) adversary without oracle access to any states.  Thus, the probability that an adversary in the CHRS model breaks the one-way puzzle construction is negligible as well.  Formally, we have the following lemma.

\begin{theorem}[One-way puzzle security]
\label{thm:owpuzz-security}
    For all adversaries $\calA^{\{\ket{\psi_l}\}}$ with access to polynomially many copies of the CHRS states, if the sampler samples a puzzle $(k = (k_1, \ldots, k_l), s)$ and inputs the puzzle $s$ to $\calA^{\{\ket{\psi_l}\}}$, then $\calA$ passes verification with negligible probability.
\end{theorem}
\begin{proof}
    Consider an adversary $\mathcal{A}^{\{\ket{\psi_l}\}_{l}}$ with samples of the common Haar random state, and $\mathcal{A}^{\{\ket{\psi'_l}\}_{l}}$ that gets samples of an independently sampled Haar random state.  We first claim that $\mathcal{A}^{\{\ket{\psi'_l}\}_{l}}$ passes verification with probability at most $(3/5)^{n}$.  Consider the following adversary $\mathcal{B}$ that makes no calls to the common Haar random state, but runs in exponential time.  $\mathcal{B}$ maintains a database of classical description of states $\{(l, \ket{\phi_l})\}$.  Whenever $\mathcal{A}$ requests copies of one of the states $\ket{\psi_l}$, if $l$ appears in the database of $\mathcal{B}$, then $\mathcal{B}$ generates a copy of $\ket{\phi_l}$, otherwise $\mathcal{B}$ samples $\ket{\phi_l}$ independently at random from the Haar measure on $l$ qubits, adds $(l, \ket{\phi_l})$ to the database, and outputs a copy of $\ket{\phi_l}$.  Because $\ket{\phi_l}$ are sampled independently at random from $\{\ket{\psi_l}\}_l$, the common Haar random state, $\mathcal{B}$ is accepted with the same probability as $\mathcal{A}^{\{\ket{\psi'_l}\}_{l}}$.  However $\mathcal{B}$ makes no queries to the common Haar random state and by \Cref{lem:no_samples_cant_verify}, passes one-way puzzle verification with probability at most $(3/5)^{n}$. Therefore $\mathcal{A}^{\{\ket{\psi_l}\}_{l}}$ passes one-way puzzle verification with probability at most $(3/5)^{n}$.
    
    Now assume for the sake of contradiction that $\mathcal{A}^{\{\ket{\psi_l}\}_{l}}$ passes verification with probability greater than $(3/5)^{n} + O(t^2 / 2^n)$.  Then consider the LOCC procedure where Alice first calls $\mathsf{Samp}$ and sends the puzzle to Bob.  Bob then runs $\mathcal{A}$ on $s$ and the copies of their state, and returns the output.  Alice then runs $\mathsf{Ver}$ and outputs $1$ if the verifier accepts.  If $\mathcal{A}^{\{\ket{\psi_l}\}_{l}}$ is accepted with probability greater than $(3/5)^n + O(t^2 / 2^n) = \negl(n)$, then the LOCC protocol described above distinguishes between the case when Alice and Bob have the same state or different state with advantage greater than $O(t^2 / 2^n)$, contradicting \Cref{lem:locc-harr-indi}. 
\end{proof}


Putting~\Cref{thm:iv-owpuzz} together with~\Cref{cor:ev-owpuzz-non-exist}, we have the following.
\begin{corollary}
    In the common Haar random state model, inefficiently verifiable one-way puzzles exist but efficiently verifiable one-way puzzles do not exist. 
\end{corollary}

\section{Separation in the Haar random swap oracle model}
\label{sec:unitary_oracle_sep}

\noindent \textbf{Note}: We highlight again that there is a bug in this section, pointed out to us by Eli Goldin and Mark Zhandry. While \Cref{lem:swap_from_samples} is true, it only holds for fixed states $\rho$, and thus can not be applied to get an oracle separation between $\EFI$ pairs and $\OWSG$ in the Haar random swap model. The construction of $\owpuzz$ and $\EFI$ can be lifted to the unitary model but the attack on $\OWSG$ cannot be lifted with the approach in this paper. Treating the swap oracle as a reflection around $\ket{0} - \ket{\psi}$ for a Haar random $\ket{\psi}$ allows one recover the result, see \cite{goldin2025translating} for a more detailed discussion.  

In the previous section, we showed that in CHRS model, one-way state generators do not exist.  In this section, we show how to lift certain constructions in the CHRS to a unitary oracle that swaps between $\ket{0}$ and a Haar random state.  Given a state $\ket{\psi}$, define the \emph{swap oracle} to be the following:
\begin{definition}[Swap oracle]
    Let $\ket{\psi}$ be an $n$-qubit state, then the \emph{swap oracle} $\oracle_{\psi}$ is defined as follows:
    \begin{equation*}
        \oracle_{\psi} = \ket{0^{n}}\!\!\bra{\psi} + \ket{\psi}\!\!\bra{0^{n}} + \id_{\perp}\,.
    \end{equation*}
    where we assume WLOG that $\ket{\psi}$ is orthogonal to $\ket{0^{n}}$, since if not, we can always append a single $\ket{1}$ to it in order to make it orthogonal.
    Here 
    $\id_{\perp}$ is the identity on the subspace orthogonal to $\mathrm{span}\{\ket{0^{n}}, \ket{\psi}\}$.  For a family of states $\Psi$, the swap oracle $\oracle_{\Psi}$ is a family of oracles that each swap around the corresponding state in the state family.
\end{definition}

We can instantiate the swap model with any common reference quantum state model, but for this paper we will define the Haar random swap oracle as follows.
\begin{definition}[Haar random swap oracle]
     Let $\{\ket{\phi_{m}}\}_{m \in \mathbb{N}}$ be a collection of states sampled from the Haar measure on $2^{m}-1$ dimensions, where $\ket{\phi_{m}}$ is a state on $m$ qubits, conditioned in being orthogonal to $\ket{0}$.%
     \footnote{
         We note that we add the condition that $\ket{\psi_{m}}$ is orthogonal to $\ket{0}$ so that the swap oracle is a unitary.  However, for all polynomials, $\poly(m)$-copies of a Haar random state on $2^{m}$-dimensions has negligible trace distance to $\poly(m)$-copies of a Haar random state on $2^{m} - 1$ dimensions, so making this change is indistinguishable to any adversary that makes $\poly(m)$ calls to the swap oracle, which is what we consider below. 
         Thus, it will often be convenient\textemdash and it will add only negligible error to our analysis\textemdash to consider the state prepared by the oracle as if it were Haar random on the full $2^{m}$-dimensional space.
     }
     The Haar random swap oracle is the collection of unitaries:
     \begin{equation*}
         \oracle_{\ket{\Haar}} = \oracle_{\{\ket{\phi_{m}}\}}\,.
     \end{equation*}
\end{definition}


Next, we show that one can use copies of the Haar random state to simulate calls to the Haar random swap oracle, at an inverse polynomial trace distance error. We note that a very similar swap oracle over \emph{two} Haar random state was considered in the work of \cite{zhandry24space}, where he proved a very similar state simulation technique.  We adapt his proof strategy for our case.

\begin{longfbox}[breakable=false, padding=1em, margin-top=1em, margin-bottom=1em]
    \begin{algorithm}\label{alg:simulating_swap_oracle}
    Algorithm for simulating swap oracle $\oracle_{\psi}$.
    \end{algorithm}
    \noindent \textbf{Input: } Unknown quantum state $\rho$ in register $\reg{R}$ and $2q(\lambda)+1$ copies of $\ket{\psi}$, orthogonal to $\ket{0^{n}}$.
    \begin{enumerate}
        \item\label{step:simulate_1} Coherently measure $\reg{R}$ using the POVM $\{\proj{0^{n}}, \id - \proj{0^{n}}\}$, saving the result in register $\reg{A}_1$.
        \begin{enumerate}
            \item\label{step:simulate_1_a} If the measurement result in $\reg{A}_1$ has outcome $\proj{0^{n}}$, swap out register $\reg{R}$ with a fresh copy of $\ket{\psi}$.  
            \item\label{step:simulate_2} If the measurement result in $\reg{A}_1$ has outcome $\id - \proj{0^{n}}$, perform a symmetric subspace projector $\Pi_{\mathrm{sym}}^{q(\lambda) + 1}$ on register $\reg{R}$ and $q(\lambda)$ other registers containing fresh copies of $\ket{\psi}$, saving the result in register $\reg{A}_2$.
            \begin{enumerate}
                \item\label{step:simulate_3} If the measurement result in $\reg{A}_2$ has outcome $\Pi_{\mathrm{sym}}$, swap out register $\reg{R}$ with an ancilla register containing the state $\ket{0^{n}}$.
            \end{enumerate}
        \end{enumerate}
        \item\label{step:simulate_4} Coherently measure register $\reg{R}$ using the POVM $\{\proj{0^{n}}, \id - \proj{0^{n}}\}$, writing the result to register $\reg{A}_2$.
        \item\label{step:simulate_5} Coherently measure register $\reg{R}$ and $\ket{\psi}^{\otimes q(\lambda)}$ using a symmetric subspace projector $\Pi_{\mathrm{sym}}^{q(\lambda) + 1}$, writing the result to register $\reg{A}_1$.
        \item Return the $\reg{R}$ register.
    \end{enumerate}
\end{longfbox}

Before proving that the algorithm works, we will need to use the following facts, one about the post-measurement state of the symmetric subspace projector and another about Haar random states.

\begin{lemma}
Let $\ket{\phi}$ be a state perpendicular to $\ket{\psi}$, then the post-measurement state after applying the symmetric subspace projector (i.e. the sum over all permutations of $l+1$ registers) and accepting on $\ket{\phi} \otimes \ket{\psi}^{\otimes l}$ is 
\begin{equation*}
    \frac{1}{l+1} \sum_{i = 0}^{l} \ket{\psi}^{\otimes i} \ket{\phi} \ket{\psi}^{\otimes l - i}
\end{equation*}
Similarly, the post-measurement sate after applying the symmetric subspace projector and rejecting is given by
\begin{equation*}
    \frac{l}{l+1}\ket{\phi} \ket{\psi}^{\otimes l} - \frac{1}{l+1}\sum_{i = 1}^{l} \ket{\phi}^{\otimes i} \ket{\phi} \ket{\psi}^{\otimes l - i}\,.
\end{equation*}
\end{lemma}
\begin{proof}
    From \cite{harrow2013church}, the symmetric subspace projector is given by
    \begin{equation*}
        \Pi^{\mathrm{sym}}_{l+1} = \frac{1}{(l+1)!}\sum_{\pi \in S_{l+1}} P_{\pi}\,.
    \end{equation*}
    Where $P_{\pi}$ acts on $n$ registers by permuting the registers.  Note that we can split $S_{l+1}$ into cosets of $S_{l}$, where the representative of the $i$'th coset is the swap $(1, i)$.  Then we have the following
    \begin{align*}
        \frac{1}{(l+1)!}\sum_{\pi \in S_{l+1}} \ket{\phi} \otimes \ket{\psi}^{\otimes l} &= \frac{1}{(l+1)!} \sum_{i \in [l+1]} P_{(1, i)}\sum_{\pi' \in S_{l}} \id \otimes P_{\pi'} \ket{\phi} \otimes \ket{\psi}^{\otimes l}\\
        &= \frac{1}{l+1} \sum_{i \in [l+1]} P_{(1, i)} \ket{\phi} \otimes \ket{\psi}^{\otimes l}\\
        &= \frac{1}{l+1} \sum_{i = 0}^{l} \ket{\psi}^{\otimes i} \ket{\phi} \ket{\psi}^{\otimes l - i}\,,
    \end{align*}
    as desired.  To compute the state after the anti-symmetric subspace projector is applied, we simply take $\id - \Pi^{\mathrm{sym}}_{l+1}$, i.e. subtract the above state from the original state.  Taking the difference yields the desired state.
\end{proof}

We adapt the following lemma from~\cite{zhandry24space}.

\begin{lemma}[Paraphrased from~\cite{zhandry24space}]
\label{lem:haar_swap_indist}
    Let $\ket{\psi}$ be an $n$-qubit state drawn from a phase invariant distribution%
    \footnote{
        A distribution on quantum states over a subspace is \emph{phase invariant} if it is invariant to applying a uniformly random phase to each basis state of the subspace~\cite{zhandry24space}. In particular, note that the Haar measure is phase invariant.
    }
    and $A^{\oracle_\psi}$ be a quantum oracle algorithm that makes $p(\lambda) = \poly(\lambda)$ many queries to the swap oracle $\oracle_{\psi}$. Let $\oracle_{\psi}^{\mathsf{res}}$ be a simulation that maintains a reservoir register $\reg{res}$ containing up to $p(\lambda)$ copies of $\ket{\psi}$, and performs the following unitary for each query:
    \begin{align*}
        \oracle^{\mathsf{res}} 
        = 
        \sum_{k \in [p(\lambda)]}
        \ket{0^{n}}_{\reg{A}}
        \ket{\psi^k}_{\reg{res}} 
        \!\!
        \bra{\psi}_{\reg{A}}
        \bra{\psi^{k-1}}_{\reg{res}}
        + 
        \ket{\psi}_{\reg{A}}
        \ket{\psi^{k-1}}_{\reg{res}} 
        \!\!
        \bra{0^{n}}_{\reg{A}}
        \bra{\psi^{k}}_{\reg{res}}
        + 
        \id_{\perp}
        \,.
    \end{align*}
    where $\reg{A}$ is the query register of the algorithm and $\reg{res}$ is the reservoir state, with $\ket{\psi^k}$ representing the the state of the reservoir containing $k$ copies of $\ket{\psi}$.

    Then we have that for all input states $\rho$,
    \begin{align*}
        \E_{\psi \gets \Haar} 
        \left[
            A^{\oracle_{\psi}}(\rho)
        \right]
        =
        \E_{\psi \gets \Haar} 
        \left[
            A^{\oracle_{\psi}^{\mathsf{res}}}(\rho)
        \right]
        \,.
    \end{align*}
\end{lemma}
\begin{proof}
    This follows directly by an application of the proofs of Lemmas 5.5 and 5.9 of~\cite{zhandry24space} to $\oracle_{\psi}$.%
    \footnote{The lemmas in~\cite{zhandry24space} are stated for a slightly different oracle. The oracle of~\cite{zhandry24space} has an index register and uses a reservoir of many copies of each state from an indexed collection of different Haar random states. Furthermore, the oracle\textemdash at least when queried on index 0\textemdash swaps in two different Haar states. Formally the lemmas do not apply to our case. However, the proofs of the lemmas do not make use of these differences with our setting and therefore apply directly to our oracle as a special case.}
\end{proof}

We note that $\oracle^{\mathsf{res}}$ uses a perfect projective measurement onto the state $\ket{\psi}$. It remains to show how to implement this using copies of the state $\ket{\psi}$.  Like~\cite{zhandry24space}, we use a technique from~\cite{Ji_2018}, implementing a projection on the symmetric subspace of polynomially many copies of $\ket{\psi}$. We show how to do this for completeness. 
We combine them in the following lemma where we extend \Cref{lem:haar_swap_indist} to an algorithm that takes as input copies of the state $\ket{\psi}$, sampled from the Haar measure.

\begin{lemma}[Swap oracle from sample access]\label{lem:swap_from_samples}
    Let $\ket{\psi}$ be an $n$-qubit Haar random quantum state and $A^{\oracle_\psi}$ be a polynomial-space quantum oracle algorithm that makes $p(\lambda) = \poly(\lambda)$ many queries to the swap oracle $\oracle_{\psi}$.  Then for every $\epsilon > 0$, there exists a polynomial-space quantum algorithm $B(\ket{\psi}^{\otimes p(\lambda)\left(\frac{12}{\epsilon}\right)}, \cdot)$ that such for all $\rho$,
    \begin{equation*}
        \td\left(\mathbb{E}_{\psi \leftarrow \Haar(2^n)}\left[B(\ket{\psi}^{\otimes p(\lambda)\left(\frac{12}{\epsilon}\right)}, \rho)\right], \mathbb{E}_{\psi \leftarrow \Haar(2^{n})} \left[A^{\oracle_{\psi}}(\rho)\right]\right) \leq \epsilon\,.
    \end{equation*}
\end{lemma} 
\begin{proof}
    $B$ proceeds by simulating $A$, where for each of the oracle queries made by $A$, it uses $(2t(\lambda)+1)$ many copies of $\ket{\psi}$ to run \Cref{alg:simulating_swap_oracle}.  Note that since $\epsilon$ is a constant and $p$ is a polynomial, $B$ runs in polynomial space.  In order to analyze the error bound with the ideal algorithm, we step through the algorithm for a pure state $\proj{\phi}$, and the result will extend by linearity.%
    \footnote{For the sake of brevity, we drop the expectation over the Haar measure and use kets, but all equations should be taken as being averaged over the Haar measure in $\ket{\psi}$.}  First note that we can always find phases such that 
    \begin{equation*}
        \ket{\phi} = \alpha_{0} \ket{0^{n}} + \alpha_{\psi} \ket{\psi} + \alpha_{\perp} \ket{\phi_{\perp}}\,.
    \end{equation*}
    Then after performing the first measurement, we have the following state mixed state:
    \begin{equation*}
        \proj{0^{n}} \ket{\phi} \otimes \ket{1}_{\reg{A}_1} + (\id - \proj{0^{n}}) \ket{\phi} \otimes \ket{0}_{\reg{A}_1} = \alpha_{0} \ket{0^{n}} \otimes \ket{1}_{\reg{A}_1} + \left(\alpha_{\psi} \ket{\psi} + \alpha_{\perp} \ket{\phi_{\perp}}\right) \otimes \ket{0}_{\reg{A}_1}\,.
    \end{equation*}

    Conditioned on $\reg{A}_1$ being $1$, we swap in one of our copies of $\ket{\psi}$, and in the $0$ branch we perform a multi-SWAP test with $q(\lambda)$ copies of the state $\proj{\psi}$ to project on the the symmetric subspace and its complement.  To save space, we analyze each branch of the superposition separately, and we will re-combine them later. From \Cref{lem:haar_swap_indist}, the state in the branch where $\reg{A}_1$ is $1$ is indistinguishable (over the Haar measure) to the following, after tracing out extra copies of$\ket{\psi}$.
    \begin{equation*}
        \alpha_0 \ket{\psi} \ket{10}_{\reg{A}_1\reg{A}_2}\,.
    \end{equation*}
    For the term with coefficient $\alpha_{\psi}$, the SWAP test will pass with probability $1$ and leave the following state
    \begin{equation*}
        \alpha_{\psi} \ket{\psi}^{t(\lambda)+1}\ket{01}_{\reg{A}_1\reg{A}_2}\,.
    \end{equation*}
    For the term $\alpha_{\perp}\ket{\phi_{\perp}}$, we apply our observation about the post-measurement state of the symmetric subspace projector to get
    \begin{multline*}
        \frac{\alpha_{\perp} }{t(\lambda) + 1} \sum_{i = 0}^{t(\lambda)} \ket{\psi}^{\otimes i} \ket{\phi_{\perp}} \ket{\psi}^{t(\lambda) - i} \ket{01}_{\reg{A}_1\reg{A}_2} \\+ \frac{\alpha_{\perp} }{t(\lambda) + 1} \left(t(\lambda)\ket{\phi_{\perp}} \ket{\psi}^{\otimes t(\lambda)} - \sum_{i = 1}^{t(\lambda)}\ket{\psi}^{\otimes i} \ket{\phi_{\perp}} \ket{\psi}^{t(\lambda) - i}  \right)\ket{00}_{\reg{A}_1\reg{A}_2}\,.
    \end{multline*}
    In the next step, for all of the branches that have $\ket{01}_{\reg{A}_1\reg{A}_2}$, we swap out the first register with a fresh ancilla containing $\ket{0^{n}}$.  Again, from \Cref{lem:haar_swap_indist}, the state is indistinguishable over the Haar measure from
    \begin{multline*}
        \alpha_0 \ket{\psi} \ket{\psi}^{\otimes t(\lambda)} \ket{10}_{\reg{A}} + \alpha_{\psi} \ket{0^{n}} \ket{\psi}^{\otimes t(\lambda)} \ket{01}_{\reg{A}} \\+ \frac{\alpha_{\perp}}{t(\lambda)+1} \left(\ket{0^{n}}\ket{\psi}^{\otimes t(\lambda)} + \sum_{i = 1}^{t(\lambda)} \ket{0} \ket{\psi}^{i-1} \ket{\phi_{\perp}} \ket{\psi}^{t(\lambda) - i}\right) \ket{01}_{\reg{A}} \\ 
        + \frac{\alpha_{\perp}}{t(\lambda) + 1}\left(t(\lambda)\ket{\phi_{\perp}} \ket{\psi}^{\otimes t(\lambda)} - \sum_{i = 1}^{t(\lambda)}\ket{\psi}^{\otimes i} \ket{\phi_{\perp}} \ket{\psi}^{t(\lambda) - i}  \right)\ket{00}_{\reg{A}}\,.
    \end{multline*}
    Re-arranging terms, we have the following state on all the registers after performing the first half of \cref{alg:simulating_swap_oracle}.
        \begin{multline*}
        \alpha_0 \ket{\psi} \ket{\psi}^{\otimes t(\lambda)} \ket{10}_{\reg{A}} + \alpha_{\psi} \ket{0^{n}} \ket{\psi}^{\otimes t(\lambda)} \ket{01}_{\reg{A}} + \frac{\alpha_{\perp}t(\lambda)}{t(\lambda)+1} \ket{\phi_{\perp}} \ket{\psi}^{\otimes t(\lambda)} \ket{00}_{\reg{A}}\\ + \frac{\alpha_{\perp}}{t(\lambda)+1} \left(\ket{0^{n}}\ket{\psi}^{\otimes t(\lambda)} + \sum_{i = 1}^{t(\lambda)} \ket{0^{n}} \ket{\psi}^{i-1} \ket{\phi_{\perp}} \ket{\psi}^{t(\lambda) - i}\right) \ket{01}_{\reg{A}} \\- \frac{\alpha_{\perp}}{t(\lambda) + 1}\left(\sum_{i = 1}^{t(\lambda)}\ket{\psi}^{\otimes i} \ket{\phi_{\perp}} \ket{\psi}^{t(\lambda) - i}  \right)\ket{00}_{\reg{A}}\,.
    \end{multline*}
    Note that the 
    trace of the absolute value 
    final line is equal to $\abs{\alpha_{\perp}}^2 (2t(\lambda)+1)/ (t(\lambda) + 1)^2$. Thus, we can instead consider the following state, remembering that we have incurred a trace distance cost of $\abs{\alpha_{\perp}}^2 (2t(\lambda)+1)/ (t(\lambda)+1)^2 \leq 2\abs{\alpha_{\perp}}^2 / (t(\lambda)+1)$.
    \begin{equation*}
        \alpha_0 \ket{\psi} \ket{\psi}^{\otimes t(\lambda)} \ket{10}_{\reg{A}} + \alpha_{\psi} \ket{0} \ket{\psi}^{\otimes t(\lambda)} \ket{01}_{\reg{A}} + \frac{\alpha_{\perp}t(\lambda)}{t(\lambda)+1} \ket{\phi_{\perp}} \ket{\psi}^{\otimes t(\lambda)} \ket{00}_{\reg{A}}\,.
    \end{equation*}
    For this state we can cleanly trace out the $t(\lambda)$ registers which contain $\ket{\psi}^{\otimes t(\lambda)}$, which means the state on the rest was close in trace distance to
    \begin{equation*}
        \ket{\phi_{\mathrm{half}}} = \alpha_0 \ket{\psi} \ket{10}_{\reg{A}} + \alpha_{\psi} \ket{0} \ket{01}_{\reg{A}} + \frac{\alpha_{\perp}t(\lambda)}{t(\lambda)+1} \ket{\phi_{\perp}} \ket{00}_{\reg{A}}\,.
    \end{equation*}
    The final half of the algorithm is the same as the first, except with the $\reg{A}_1$ and $\reg{A}_2$ registers swapped, so the state after the algorithm is done will be close to the following state.
    \begin{equation*}
        \ket{\phi_{\mathrm{final}}} = \alpha_0 \ket{\psi}_{\reg{A}} + \alpha_{\psi} \ket{0} + \frac{\alpha_{\perp}t(\lambda)^2}{(t(\lambda)+1)^2} \ket{\phi_{\perp}}\,.
    \end{equation*}
    By the same argument as before, the trace distance between the second half of the algorithm acting on $\ket{\phi_{\mathrm{half}}}$ and $\ket{\phi_{\mathrm{final}}}$ is at most
    \begin{equation*}
        \frac{2\abs{\alpha_{\perp}t(\lambda) / (t(\lambda)+1)}^2 t(\lambda)}{(t(\lambda)+1)^2} \leq \frac{2\abs{\alpha_{\perp}}^2}{t(\lambda)+1}\,.
    \end{equation*}
    Here we used the fact that the new amplitude of $\ket{\phi_{\perp}}$ is $\alpha_\perp t(\lambda) / (t(\lambda)+1)$ and plugged it into the trace distance calculation from before.  By the triangle inequality, the trace distance between the actual state of the algorithm and $\ket{\phi_{\mathrm{final}}}$ is at most
    \begin{equation*}
        \td(\proj{\phi_{\mathrm{final}}}, \rho_{\mathrm{alg}}) \leq \frac{2\abs{\alpha_{\perp}}^2}{t(\lambda)+1} + \frac{2\abs{\alpha_{\perp}}^2}{t(\lambda)+1} \leq \frac{4\abs{\alpha_{\perp}}^2}{t(\lambda)+1} \leq \frac{4}{t(\lambda)+1}\,.
    \end{equation*}
    Finally, we can write down the state after the ideal swap as
    \begin{equation*}
        \ket{\phi_{\mathrm{ideal}}} = \alpha_0 \ket{\psi} + \alpha_{\psi} \ket{0} + \alpha_{\perp}\ket{\phi_{\perp}}\,.
    \end{equation*}
    We can compute directly the trace distance between the two states to get the following bound
    \begin{equation*}
        \td(\proj{\phi_{\mathrm{final}}}, \proj{\phi_{\mathrm{ideal}}}) \leq \frac{2t(\lambda) + 1}{(t+1)^2} \leq \frac{2}{t(\lambda)+1}\,.
    \end{equation*}
    Applying the triangle inequality for trace distance, we get that the trace distance between the state of the algorithm on the first register and the ideal state is upper bounded by
    \begin{equation*}
        \td(\rho_{\mathrm{alg}}, \proj{\phi_{\mathrm{ideal}}}) \leq \frac{6}{t(\lambda)+1}\,.
    \end{equation*}
    Setting $t(\lambda) = \frac{6p(\lambda)}{\epsilon} - 1$, we get a trace distance error of at most $\epsilon / p(\lambda)$.  For an algorithm that makes $p(\lambda)$ calls to the oracle, we apply the triangle inequality to every call to get a total error bound of $\epsilon$ over the course of the entire algorithm $A^{\oracle_{\psi}}$.  This means that the algorithm requires $\frac{12p(\lambda)}{\epsilon}$ copies of the state $\ket{\psi}$, as desired. Finally, we note that since this bound holds for all pure states, we can write all mixed states as a mixture of pure states and apply this to each of them, so this applies to all entangled inputs too. 
\end{proof}

\begin{remark}
    \Cref{lem:haar_swap_indist,lem:swap_from_samples} allow us to simulate the effect of querying the Haar random swap oracle using only polynomially many copies of the Haar random state itself, up to inverse polynomial error. 
    Note that in order to apply it directly to an interactive computation (for instance, an interaction that occurs as part of a cryptographic security game), we must have a single global simulator that simulates the oracle for all the parties to the computation.
    For instance, consider a scenario in which two parties prepare a state close to $\ket{\psi} + \ket{1}$ by querying the Haar random swap oracle on the state $\ket{0} + \ket{1}$ and then proceed to query a second time expecting $\ket{0} + \ket{1}$ on the output.
    If there were only a single global simulator, the second query would approximately undo the first, and the state is guaranteed by \Cref{lem:haar_swap_indist,lem:swap_from_samples} to return to $\ket{0} + \ket{1}$ up to inverse polynomial error.
    If, on the other hand, the two queries are made to two \emph{different} simulators (say a different simulator for each party), the second query does not undo the first. This is because this process returns a copy of $\ket{\psi}$ that was pulled from the reservoir of the first simulator to the reservior of the second simulator, and thus produces on the output a mixed state that is entangled with both simulators and potentially quite far from $\ket{0} + \ket{1}$.%
    \footnote{
        We thank Mark Zhandry and Eli Goldin for pointing out this scenario.
    } 
    This is, however, not an issue for our case. As we will see below, in order to upgrade the existence of the primitives we are interested in to the Haar random swap oracle setting, we only need to consider a simulator for the adversary, and not the challenger of the security game. Similarly, in order to rule out efficiently verifiable $\owpuzz$ and $\OWSG$ in the Haar random swap oracle setting, we only need to simulate the verification algorithm.
\end{remark}

With these lemmas, we can show that since $\EFI$ pairs exist in the common Haar random state model, they also exist in the corresponding swap model. We apply the result to the common Haar random state model to get the following corollaries. 
\begin{corollary}\label{cor:efi_swap_exists}
    $\EFI$ pairs exist relative to $(\oracle_{\ket{\Haar}}, \mathsf{UnitaryPSPACE})$.
\end{corollary}
\begin{proof}
    The construction of $\EFI$ pairs runs the construction in the common Haar random state model, using the swap oracle to generate copies of the desired states.  To show security of the model, we show that any adversary in the Haar random swap model would imply an adversary against the $\EFI$ pair in the common Haar random state model. 

    Assume for the sake of contradiction that there is an adversary $A$ that breaks the security of the $\EFI$ pair relative to $(\oracle_{\ket{\Haar}}, \mathsf{UnitaryPSPACE})$.  Consider an algorithm $\mathrm{Sim}$ that, given bit $b$, first runs the construction of the $\EFI$ pair using the swap oracle with bit $b$, and then runs $A$ on the output.  
    By assumption, the adversary $A$ breaks the security of the $\EFI$ pair, so there exists a polynomial $q$ such that 
    \begin{equation*}
        \left| \Pr[\mathrm{Sim}(0) = \top] - \Pr[\mathrm{Sim}(1) = \top]\right| \geq \frac{1}{q(\lambda)}\,.
    \end{equation*}
    Note that this algorithm starts from the all $0$'s state, so we can apply \Cref{lem:swap_from_samples} to simulate it.  When we apply \Cref{lem:swap_from_samples} with $\epsilon = 1 / (2q(\lambda))$, we will get a simulator $\widetilde{\mathrm{Sim}}$ in the common Haar random state model.  We assume without loss of generality that the construction of $\EFI$ pairs begins by generating all of the copies of the Haar random states by calling the swap oracle on $\ket{0}$, so we see that \Cref{alg:simulating_swap_oracle} will simulate the $\EFI$ pair perfectly (as in, it will swap the initial state $\ket{0}$ with copies of $\ket{\psi}$ and then terminate, and the construction never calls the swap oracle again).  Thus, the probability that $\widetilde{\mathrm{Sim}}(b)$ accepts is exactly the probability that an adversary $\widetilde{A}$ (the simulation of $A$) accepts when given $\rho_{b}$, the corresponding state in the $\EFI$ pair.  

    Thus, the simulated adversary breaks the indistinguishability of the $\EFI$ pair with advantage $1 / (2q(\lambda))$, which contradicts the security of the $\EFI$ pair in the common Haar random model.  Thus, we conclude that the construction of $\EFI$ pairs in the Haar random swap model is secure.  
\end{proof}

We can apply the same idea to $\onePRS$ and $\owpuzz$ (and more generally, any primitive with a non-interactive security game) and conclude that they exist in the common Haar random swap model. We claim that one-way state generators \emph{do not} exist relative to the Haar random swap oracle model.

\textbf{As we pointed out at the start of the section, the proof of the claim is wrong. See \cite{goldin2025translating} for more details.}
\begin{claim}\label{cor:owsg_swap_doesnt}
    Relative to $(\oracle_{\ket{\Haar}}, \mathsf{UnitaryPSPACE})$, efficiently verifiable one-way state generators do not exist.
\end{claim}
\begin{proof}
    We consider a similar quantum threshold attack as in \Cref{thm:one_way_state_generators_do_not_exist}.  Let $(\mathsf{KeyGen}, \mathsf{StateGen}, \mathsf{Ver})$ be a candidate one-way state generator.  Applying \Cref{lem:swap_from_samples} with $\epsilon = 1/5$, we can construct a simulated verifier $\widetilde{\mathsf{Ver}}$ which takes $\poly(\lambda)$ copies of the Haar random state from the swap oracle, and satisfies the following properties:
    \begin{itemize}
        \item On average over the Haar random states, when $\widetilde{\mathsf{Ver}}$ is given $k$ and a copy of $\rho_k$ (where $\rho_k$ is the real one-way state generator state), $\widetilde{\mathsf{Ver}}$ accepts with probability at least $3/4$.
        \item If $\widetilde{\mathsf{Ver}}$ accepts with probability $p$, then $\mathsf{Ver}$ accepts with probability at least $p - 1/5$.  
    \end{itemize}
    Let $U_k$ be the unitary that implements $\widetilde{\mathsf{Ver}}$ on key $k$, which takes in a copy of a one-way state generator state in register $\reg{A}$, and $\poly(\lambda)$ copies if the Haar random state from the swap oracle in register $\reg{B}$, and let $\proj{\geq m}$ be the projector onto bit strings with Hamming weight greater than $m$, and consider the following measurement.
    \begin{equation*}
        \Pi_k = \left(\left(U_k^{\dagger}\right)^{\otimes 10\lambda}_{\reg{CD}} \left(\proj{\geq 5\lambda}_{\reg{C}} \otimes \id_\reg{D}\right) \left(U_k\right)^{\otimes 10\lambda}_{\reg{AB}}\right)\,.
    \end{equation*}
    Here $\reg{C}$ is the output register of $U_k$, and $\reg{D}$ is the purifying register for the output.
    \begin{claim}\label{claim:owsg_swap_completeness}
        On input $(\rho_k \otimes \proj{\phi_1} \otimes \ldots \otimes \proj{\phi_1})^{\otimes 10\lambda}$, $\Pi_k$ accepts with probability at least 
        \begin{equation*}
            1 - \exp\left(-\lambda/8\right)\,.
        \end{equation*}
    \end{claim}
    \begin{proof}
        From the first observation, every run of $U_k$ has probability at least $3/4$ of accepting, and since the input state is in tensor product on the original state, every run is independent of each other.  Applying Hoeffdings inequality, we get the following:
        \begin{equation*}
            \Pr[\Pi_k\ \mathrm{rejects}] \leq \exp\left(-2\frac{(\lambda/4)^2}{n}\right) = \exp\left(-\lambda/8\right)\,.
        \end{equation*}
        This completes the proof of \Cref{claim:owsg_swap_completeness}.
    \end{proof}

    We also prove the other direction in a similar way to \Cref{claim:owsg_soundness}.

    \begin{claim}
    \label{claim:owsg_swap_soundness}
        If $\Pi_k$ accepts with probability $\geq 1/3$ for some key $k$, then the probability that the \emph{original} verification accepts is at least $3/4$.
    \end{claim}
    \begin{proof}
        From the same proof as \Cref{claim:owsg_soundness}, the probability that the simulated verifier accepts is at least $1 - \frac{1}{5\lambda}$.  By \Cref{lem:swap_from_samples}, the probability that the original verifier would accept is at least $1 - \frac{1}{5\lambda} - \frac{1}{5} \geq \frac{3}{4}$ for suitably large $\lambda$. 
    \end{proof}
    The algorithm for breaking the one-way state generator is to first query the Haar random swap oracle $\poly(\lambda)$ many times to produce that many copies of the corresponding Haar random state, and then use the $\mathsf{UnitaryPSPACE}$ oracle to 
    run threshold search on those $\poly(\lambda)$ copies of the Haar random state and 
    $O(\lambda^2)$ many copies of the $\OWSG$ state.  Note that, as before, the verification might query the Haar random swap oracle for many different integers $m$, so our attack must make $
    \poly(\lambda)$ copies of the Haar random state for all $m$ from $1$ to $\poly(\lambda)$ in order to run the simulated verification.  From the previous claims, the procedure is guaranteed to output a key that passes verification with probability $3/4$, with constant probability.  
\end{proof}
By the same observation as in the common Haar random state case, we note that the exact same attack rules out constructions of efficiently verifiable one-way puzzles in the Haar random swap model.

\begin{theorem}
    Relative to $(\oracle_{\ket{\Haar}}, \mathsf{UnitaryPSPACE})$, $\EFI$ pairs, $\onePRS$, and $\owpuzz$ exist, but one-way state generators do not. 
\end{theorem}

Thus, we have achieved a unitary oracle that separates one-way state generators from $\EFI$ pairs, $\onePRS$, and $\owpuzz$.  
Note that the oracle we provide is a randomized unitary oracle\textemdash that is, it is a unitary oracle that is chosen from a probability distribution. However, we note that this is not a disadvantage, as if we desire a fixed deterministic unitary oracle, this can easily be achieved as well by applying the techniques of~\cite[proof of Theorem 1.1]{aaronson2007quantum} to $\oracle_{\ket{\Haar}}$.

\bibliographystyle{plain}
\bibliography{refs}

\end{document}